\documentclass{article}
\usepackage{graphicx} 

\title{Multidimensional multiplicative Poisson vertex algebras}
\author{Pengfei Yang\and Matteo Casati}
\date{}
\RequirePackage{amsmath}
\RequirePackage{libertine}

\RequirePackage[slantedGreek,libertine]{newtxmath}

\RequirePackage{mathtools}

\RequirePackage[standard,thmmarks,amsmath]{ntheorem}

\RequirePackage{bm}
\RequirePackage{array} 
\RequirePackage{siunitx} 
\RequirePackage{upgreek} 
\RequirePackage{calc} 
\usepackage{enumitem}

\usepackage[numbers]{natbib}

\theoremsymbol{}
\renewtheorem{theorem}{Theorem}
\renewtheorem{proposition}{Proposition}
\theorembodyfont{\normalfont}
\renewtheorem{definition}{Definition}
\renewtheorem{lemma}{Lemma}
\renewtheorem{remark}{Remark}

\def\Z{\mathbb{Z}}	
\def\C{\mathbb{C}}	
	
\def\cS{\mathcal{S}}
\newcommand{\mc}[1]{\mathcal{#1}}

\def\mlambda{\boldsymbol{\lambda}}
\def\mmu{\boldsymbol{\mu}}

\def\ud{\mathrm{d}}	
\newcommand{\dev}{\partial}

\def\A{\mathcal{A}}

\def\F{\mathcal{F}}

\begin{document}

\maketitle
\begin{abstract}
In this paper we introduce the notion of multidimensional multiplicative Poisson vertex algebra, the generalization of the notion of multiplicative Poisson vertex algebra to a difference algebra endowed with $D$ commuting shifts. After showing the equivalence of this notion to the notion of Hamiltonian difference operator on a $D$-dimensional lattice, we characterize scalar local Hamiltonian difference operators up to the order $(-2,2)$ and investigate the bi-Hamiltonian pairs they form.
\end{abstract}
\section{Introduction}

Evolutionary differential-difference equations (D$\Delta$Es) are a class of systems for functions depending on two sets of variables, of which one is a continuous parameter (the time) and the other takes value on a lattice. The prototypical example of such class is the
(infinite) Volterra lattice, satisfied by a function $u(n, t) (n \in \mathbb{Z}, t \in \mathbb{R})$ such that
\begin{equation}
    \partial_tu(n,t)=u(n,t)(u(n+1,t)-u(n-1,t)).
\end{equation}
By denoting $u(n,t):=u_0=u , u(n + m, t) := u_m$, we can write the Volterra equation as
\begin{equation}
    \partial_tu=u(u_1-u_{-1}).
\end{equation}
Moreover, introducing the shift operator $\cS f(u, u_1,\dots, u_n) = f (u_1, u_2,\dots, u_{n+1})$, i.e. $\cS u_m=u_{m+1}$, this equation can be cast in Hamiltonian form for the Hamiltonian difference operator
\begin{equation}
    K=uu_1\cS-uu_{-1}\cS^{-1}
\end{equation}
and the functional $H=\int u$, so that $\delta H=1$.

More in general, let us consider the class of evolutionary equation for a set of functions $u^i(n,t)$, $i\in \{1,\dots , \mathscr{l}\}$: we call  $\mathscr{l}$ the number of \textit{components} of the system. A generic \emph{evolutionary} D$\Delta$E system for $\mathscr{l}$ functions of one (discrete) spatial and one (continuous) time variable is of the form $(u=(u^1,\dots, u^\mathscr{l}))$
\begin{equation}
    \frac{\partial}{\partial t}u^i(n,t)=F^i(\dots , \cS^{-1}u,u,\cS u,\dots )
\end{equation}
for $i=1,\dots ,\mathscr{l}$. Among this class of systems, we are particularly interested to those who can be cast in a local \textit{Hamiltonian form}, namely written as
\begin{equation}
    \partial_tu^i=\sum_{j=1}^{\mathscr{l}}K^{ij}\left(\frac{\delta H}{\delta u^j}\right)
\end{equation}
for a difference operator $K$ and a local functional $H$. The operator $K$ must be such that on the one hand it defines a Lie algebra structure (the \textit{local} Poisson bracket) on the space of local functionals, and on the other hand establishes the usual Lie algebra morphism between local functionals and the Lie
algebra of evolutionary vector fields; if this is the case, we say that $K$ is a \textit{Hamiltonian operator}.

The study of Hamiltonian operators is particularly important in the theory of integrable systems and in deformation quantisation. It is well known, for instance, that Magri \cite{magri} introduced the concept of compatible pair of Poisson brackets (the bi-Hamiltonian structure) and related it to the complete integrability of systems of partial differential equations.

The foundations of calculus for difference operators have been developed by Kuperschmidt \cite{K85}. The Hamiltonian structures
of many integrable differential-difference systems have been identified (see \cite{WangJP} and references therein).
\paragraph{Hamiltonian equations and Poisson vertex algebras.}
The study of the Hamiltonian structures of D$\Delta$Es is the main topic of this paper. The theory of multiplicative Poisson vertex algebras (PVAs) \cite{10.1093/imrn/rny242,LocalandNon-local} has provided a convenient
framework for their description. Using this formalism, a classification of $\mathscr{l}=1$ components (or scalar) difference Hamiltonian operators up to the 5th order was obtained in \cite{10.1093/imrn/rny242}; multi-component first orders operator were first studied in \cite{D89}; the classification for the 1st order, $\ell=2$ case was recently completed in \cite{CV25}.

The basic idea of the theory of PVAs, originally developed in \cite{BdSK09} to provide an algebraic framework to the theory of Hamiltonian PDEs, is endowing the space of \emph{densities} of local functionals (namely, differential polynomials in the continuous case, difference polynomials in the semi-discrete case we consider in this paper) with a bilinear operation, the $\lambda$ \emph{bracket}, which is then used to define both the Poisson bracket of local functionals and the Hamiltonian evolutionary vector fields. It can be proved that the properties enjoyed by the $\lambda$ bracket make it equivalent to the notion of Hamiltonian operator, with the advantage that the theory provides a convenient, closed form algorithm for all the standard computations, called the \emph{master formula}.

A local multiplicative Poisson Vertex Algebra \cite{10.1093/imrn/rny242} is a commutative associative algebra $\mc A$ with an automorphism $\cS$, endowed with a bilinear operation $\{-_{\lambda}-\}\colon\A\times\A\to\A[\lambda,\lambda^{-1}]$ satisfying the following set of axioms:
\begin{enumerate}
\item $\{f_{\lambda}\cS(g)\}=\lambda\,\cS(\{f_{\lambda} g\})$, $\{\cS(f)_{\lambda}g\}=\lambda^{-1}\{f_{\lambda} g\}$ (sesquilinearity)
\item $\{f_{\lambda}gh\}=\{f_{\lambda}g\}h+g\{f_{\lambda}h\}$, $\{fg_{\lambda}h\}=\{f_{\lambda\cS}h\}g+\{g_{\lambda \cS}h\}f$ (Leibniz property)
\item $\{g_{\lambda}f\}=-{}_\to\{f_{(\lambda\cS)^{-1}}g\}$ (PVA-skewsymmetry)
\item $\{f_{\lambda}\{g_{\mu}h\}\}-\{g_{\mu}\{f_{\lambda}h\}\}=\{\{f_{\lambda}g\}_{\lambda\mu}h\}$ (PVA-Jacobi identity).
\end{enumerate}
To understand the notation, first we observe that the $\lambda$ bracket takes value in the Laurent polynomials of an indeterminate variable $\lambda$ with coefficients in $\A$. Such coefficients depend on the arguments $f$ and $g$, so we can write
$$
\{f_{\lambda}g\}=\sum_{s=M}^NB_{s}(f,g)\lambda^s,\qquad M,N\in\Z.$$
The automorphism $\cS$ acts trivially on the variable $\lambda$, so that, for example, the sesquilineraity property should be interpreted as
$$
\sum_s B_s(f,\cS g)=\sum_s \cS \left(B_s(f,g)\right)\lambda^{s+1},\qquad \sum_s B_s(\cS f, g)=\sum_s B_s(f,g)\lambda^{s-1}.
$$
This property makes the automorphism $\cS$ an automorphism of the bracket too, namely $\{\cS f{}_\lambda\cS g\}=\cS\left(\{f_{\lambda}g\}\right)$. To interpret the notation we use for the Leibniz and the skewsymmetry property, observe that $\{f_{\lambda\cS}g\}h$ can be expanded as $\sum_s B_s(f,g)\lambda^s\cS^{s}h$; adopting the convention that the automorphism $\cS$ acts on the left, this is equivalent to $\sum_s \left[B_s(f,g)\cS^sh\right]\lambda^s$. Conversely, with ${}_\to\{f_{\lambda\cS}g\}$ we denote an expansion of the bracket where the shift operator acts on the coefficients, i.e.
$$
{}_\to\{f_{(\lambda\cS)^{-1}}g\}=\sum_s \left(\lambda\cS\right)^{-s}B_s(f,g).
$$
The key point of the theory of multiplicative PVA is the observation that the $\lambda$ bracket defines a Lie algebra bracket on the space $\mc F:=\A/\triangle A$, where $\triangle f=\left(\cS-1\right)f$, by
\begin{equation}\label{eq:PoissBrack-1Dcase}
\{\smallint f,\smallint g\}:=\int\left\{f_{\lambda}g\right\}\Big|_{\lambda=1}
\end{equation}
(here, $\smallint$ denotes the projection $\A\to\mc F$), and a Lie algebra morphism between $(\mc F,\{-,-\})$ and $(\mathrm{Der}^{\cS}\A,[-,-])$, the Lie algebra of the derivations of $\A$ commuting with the automorphism $\cS$. Taking as $\A$ an algebra of difference functions, and interpreting $\cS$ as the shift operator, $\mc F$ can be regarded as the formal space of local functionals and the Lie algebras we define are exactly the Hamiltonian structure for evolutionary D$\Delta$Es.

\paragraph{Main results and organization of the paper.}
In this paper, we extend the theory of multiplicative Poisson Vertex Algebra to the multidimensional case, namely to systems depending on more than one discrete variable (see Theorem~\ref{theorem1}). After developing the required theoretical foundations of the theory, we classify two-dimensional scalar local Hamiltonian difference operator up to the order $(-2,2)$ and investigate the bi-Hamiltonian pairs they form.

We obtain a more rigid classification than the one computed in \cite{10.1093/imrn/rny242} for scalar operators, summarized in Proposition~\ref{thm:ham1-2d} and \ref{prop:2D}. Indeed, if we consider only essentially multi-dimensional operators (those with shifts in both the spatial directions), the only admissible Hamiltonian ones are of the form
$$
P=f(u)\left(\sum_{i,j=1}^2\alpha^{ij}\left(\cS_i\cS_j-(\cS_i\cS_j)^{-1}\right)+\sum_{i=1}^2\beta^i(\cS_i-\cS_i^{-1})\right)\circ f(u),
$$
with $(\alpha^{ij},\beta^i)$ constant parameters. If, moreover, $f(u)$ is a nonvanishing function, then there exist a transformation of the coordinate $u$ such that the operators are constant.

This rigidity impacts on classification of the bi-Hamiltonian pairs, that we present in Proposition~\ref{prop:biHamNN}, \ref{prop:biHamN1}, and \ref{prop:biHam11}. If we disregard the essentially one-dimensional case (that is studied, as bi-Hamiltonian pairs of scalar one-dimensional operators, in \cite{10.1093/imrn/rny242}), the only compatible bi-Hamiltonian pairs are those for which the functions $f(u)$ are proportional, i.e. the bi-Hamiltonian pair is formed by two copies of the same operator, for different values of the constants. 
\paragraph{}
The paper is structured as follows: in Section 2 we present the definition of multidimensional multiplicative Poisson vertex algebra, in particular proving the existence of a master formula (in the sense of \cite{BdSK09}, \cite{C15}, and \cite{10.1093/imrn/rny242} -- respectively for ordinary PVAs, multidmensional PVAs, and multiplicative PVAs) and the relation between PVAs and Hamiltonian structures. In Section 3 we study multidimensional Hamiltonian difference operators, explicitly providing the classification of scalar 2-dimensional ones, up to the order $(-2,2)$. It should be noticed that our results extend in a straightforward way to the $D>2$ case. In Section 4 we consider pairs of Hamiltonian operators of order up to 2 and study their compatibility.

\section{Multidimensional multiplicative PVAs}
\subsection{Definition}
Let $(\mathcal{A},\cS_1,\cS_2,\ldots, \cS_D)$ be an algebra of difference functions in $\ell$ generators $\{u^i\}$, $i=1,\ldots,\ell$, endowed with $D$ commuting automorphisms $\{\cS_1,\ldots,\cS_D\}=:\boldsymbol{\cS}$. This algebra is an extension of the algebra of difference polynomials, i.e. the polynomial algebra over $\C$ in the variables $u^i_{(n_1,\ldots,n_D)}$, $i=1,\ldots,\ell$ and $n_\alpha\in\Z$, where
$$
\cS_\alpha u^i_{(n_1,\ldots,n_\alpha,\ldots,n_D)}= u^i_{(n_1,\ldots,n_\alpha+1,\ldots,n_D)}.
$$
For compactness, we adopt a multi-index notation, with upper case multi-indices $N=(n_1,n_2,\ldots,n_D)\in\Z^D$; we shall denote $E_\alpha\in\Z^D$ the vector with one in the $\alpha$-th position and 0 elsewehere, so that the previous expression becomes $\cS_\alpha u^i_N=u^i_{N+E_\alpha}$. We denote the order of the multi-index $N$ as $|N|=n_1+n_2+\cdots +n_D$.

This notion of difference algebra is a straightforward generalization to the case $D>1$ of the algebra of difference functions presented in \cite{DSKVW18-1,CW20}, and corresponds to the one introduced by Kuperschmidt in his discrete calculus of variations \cite{K85}. Note that the algebra $\A$ is endowed with commuting partial derivations $\frac{\dev}{\dev u^i_M}$ and that the partial derivatives of an element of $\A$ are vanishing for all but a finite number of (pairs of) indices $(i,M)$. The commutation relation between the shift operators and the partial derivatives is (see \cite{K85})
\begin{equation}\label{eq:commshift}
\frac{\dev }{\dev u^i_M}\cS_\alpha f=\cS_\alpha\frac{\dev f}{\dev u^i_{M-E_\alpha}}.
\end{equation}

\begin{definition}[$\lambda$-bracket]
    A local (multiplicative) $\lambda$-bracket (of rank D) on $\mathcal{A}$ is a $\mathbb{C}$-linear map
     \begin{equation*}
        \begin{aligned}
         {\{\cdot _{\boldsymbol{\lambda}}\cdot \}}:\mathcal{A}\times\mathcal{A}\rightarrow \mathbb{C}[\lambda^{\pm1}_1 &,\dots,\lambda^{\pm1}_D]\otimes \mathcal{A}\\
         (f,g)\mapsto {\{f_{\boldsymbol{\lambda}}g\}}&
        \end{aligned}
     \end{equation*} 
    which is \textit{sesquilinear},namely
    \begin{subequations}
        \begin{align}
        \{\cS_{\alpha}f_{\boldsymbol{\lambda}}g\}&=\lambda^{-1}_\alpha\{f_{\boldsymbol{\lambda}}g\}\label{1.1a}\\
        \{f_{\boldsymbol{\lambda}}\cS_{\alpha}g\}&=\lambda_{\alpha}\cS_{\alpha}\{f_{\boldsymbol{\lambda}}g\}\label{1.1b}
        \end{align}
    \end{subequations}
    and obeys, respectively, the \textit{right} and \textit{left Leibniz rule}
    \begin{subequations}
        \begin{align}
        \{f_{\boldsymbol{\lambda}}gh\}&=\{f_{\boldsymbol{\lambda}}g\}h+\{f_{\boldsymbol{\lambda}}h\}g \label{1.2a}\\
        \{fg_{\boldsymbol{\lambda}}h\}&=\{f_{\boldsymbol{\lambda \cS}}h\}g+\{g_{\boldsymbol{\lambda \cS}}h\}f .\label{1.2b}
        \end{align}
    \end{subequations}
\end{definition}

    By definition, the $\lambda$-bracket of two elements in $\A$ is a polynomial in $\lambda_1, \ldots, \lambda_D$ and their multiplicative inverse (we will often refer to the collection of $\lambda_\alpha$ as $\boldsymbol{\lambda}$ ) 
    with coefficients in $\A$. In general, we can write $\left\{f_\lambda g\right\}=a_{i_1, \ldots, i_D}(f, g) \lambda_1^{i_1} \ldots \lambda_D^{i_D}$, implying the summation over the repeated indices. Using the usual multi-index notation, this is equivalent to writing $a_I(f, g) \boldsymbol{\lambda}^I$. When, as in (\ref{3}), we write $\left\{f_{\boldsymbol{\lambda \cS}} g\right\}$, it means that the $\lambda$ product is $a_I(f, g)(\boldsymbol{\lambda \cS})^I$, 
    with the shift operators acting on the right (if nothing is written on the right, it is equivalent to the shift operators acting on 1 and thus the only term not vanishing is $\boldsymbol{\lambda}^I$).

\begin{definition}[Multidimensional Poisson Vertex Algebra]
 A local ($D$ dimensional) \textit{multiplicative Poisson Vertex Algebra} is a difference algebra $\mathcal{A}$ endowed with a multiplicative $\lambda$-bracket of rank $D$ which is \textit{skewsymmetric}
\begin{equation}
    \left\{g_{\boldsymbol{\lambda}} f\right\}=-{}_{\rightarrow}\{f_{{(\boldsymbol{\lambda \cS})}^{-1}} g\}\label{3}
\end{equation}
and satisfy the PVA-Jacobi identity
\begin{equation}
    \left\{f_{\boldsymbol{\lambda}}\left\{g_{\boldsymbol{\mu}} h\right\}\right\}-\left\{g_{\boldsymbol{\mu}}\left\{f_{\boldsymbol{\lambda}} h\right\}\right\}=\{\left\{f_{\boldsymbol{\lambda}} g\right\}_{\boldsymbol{\lambda \mu}} h\}.\label{Jacobi-identity}
\end{equation}

The notation used in (\ref{3}) means that the difference operators ($\boldsymbol{\lambda \cS}$) must be regarded as acting on the coefficient of the bracket, too; namely ${}_{\rightarrow}\left\{f_{\boldsymbol{\lambda \cS}} g\right\}$$=(\boldsymbol{\lambda \cS})^I a_I(f, g)$.
\end{definition}
\subsection{Master formula}
The $\lambda$-bracket between two generic elements of $\A$ can be expressed in terms of the $\lambda$-bracket between the so-called generators of the difference algebra, namely the variables $\left\{u^i\right\}$, ($i=1,\ldots,\ell$: we denote this index set as $I$). The explicit formula, which generalizes to $D$ dimensions the one first given in \cite{DSKVW18-1}, is called the \emph{master formula} for its role in all the computations involving $\lambda$-brackets.
\begin{theorem}[Master formula]\label{theorem1}
 Let $(f, g) \in \mathcal{A}$ and let $\{u^i\}_{i\in I}$ be the generators of a PVA. Then we have
\begin{equation}
    \{f_{\boldsymbol{\lambda}} g\}=\sum_{\substack{i, j\in I \\ M, N \in \mathbb{Z}^D}} \frac{\partial g}{\partial u_N^j}(\boldsymbol{\lambda \cS})^N\{u_{\boldsymbol{\lambda \cS}}^i u^j\}(\boldsymbol{\lambda \cS})^{-M} \frac{\partial f}{\partial u_M^i}.\label{master}
\end{equation}   

    In particular, the skewsymmetry and the PVA-Jacobi property hold if and only if the same properties for the generators hold:
\begin{equation}
    \{u^i_{\boldsymbol{\lambda}}\{u^j_{\boldsymbol{\mu}}u^k\}\}-\{u^j_{\boldsymbol{\mu}}\{u^i_{\boldsymbol{\lambda}}u^k\}\}=\{{\{u^i_{\boldsymbol{\lambda}}u^j\}}_{\boldsymbol{\lambda \mu}}u^k\}.\label{Jacobi-generators}
\end{equation}
\end{theorem}
\begin{proof}
    Let us denote the $\lambda$-bracket of generators as $\{u^i_{\boldsymbol{\lambda}}u^j\}=\sum_{T\in \mathbb{Z}^D} P^{ji}_{T}\mlambda^{T}$.
    The skewsymmetry of the bracket between generators is written as $\left\{u^i_{\boldsymbol{\lambda}} u^j\right\}=-{}_{\rightarrow}\{u^j_{{(\boldsymbol{\lambda \cS})}^{-1}} u^i\}$ and, after we expand the RHS as 
    \begin{equation*}
        {}_\rightarrow\{u^i_{\boldsymbol{\lambda \cS}}u^j\}=\sum_{T\in \mathbb{Z}^D}(\boldsymbol{\lambda \cS})^{T}P^{ji}_{T}=\sum_{t_i \in \mathbb{Z}}(\lambda_1 \cS_1)^{t_1}(\lambda_2 \cS_2)^{t_2}\dots(\lambda_D \cS_D)^{t_D}P^{ji}_{t_1 t_2\dots t_D}
    \end{equation*}
is equivalent to
\begin{equation}\label{eq:skewgen}
    P^{ji}_{-T}=-{\boldsymbol{\cS}}^{-T}P^{ij}_{T}.
\end{equation}

Note that the sum in Equation \eqref{master} is finite, so that it gives a well defined $\mathbb{C}$-linear map $\{-_{\boldsymbol{\lambda}}-\}: \mathcal{A} \times \mathcal{A} \rightarrow \mathbb{C}[\boldsymbol{\lambda}] \otimes \mathcal{A}$. Moreover, for $f=u^i, g=u^j, i, j \in I$, such map clearly reduces to the given Laurent polynomials $\sum_TP^{ji}_T\mlambda^T \in \mathbb{C}[\boldsymbol{\lambda}] \otimes \mathcal{A}$.

More generally, for $f=u^i_M, g=u^j_N$, equation \eqref{master} reduces to

\begin{equation}\label{2.2}
    \{{u^i_M}_{\boldsymbol{\lambda}} u^j_N\}=\boldsymbol{\lambda}^{-M}(\boldsymbol{\lambda \cS})^N\{u^i_{\boldsymbol{\lambda}} u^j\}=\sum_{T\in\Z^D}\cS^NP^{ji}_T\mlambda^{T+N-M}.
\end{equation}

It is also useful to rewrite equation \eqref{master} in the following equivalent forms, which can be checked directly:
\begin{equation}\label{6}
    \{f_{\boldsymbol{\lambda}} g\}=\sum_{j \in I, N \in \mathbb{Z}^D} \frac{\partial g}{\partial u^j_N}(\boldsymbol{\lambda \cS})^N\{f_{\boldsymbol{\lambda}} u^j\}=\sum_{i \in I, M \in \mathbb{Z}^D}{\left\{u^i_{\boldsymbol{\lambda \cS}} g\right\}}_{\rightarrow}(\boldsymbol{\lambda \cS})^{-M}\frac{\partial f}{\partial u^i_M}.
\end{equation}

We can now explicitly check that the bracket defined by the master formula \eqref{master} satisfies the properties of a $\lambda$-bracket. For the sesquilinearity property \eqref{1.1a}, \eqref{1.1b} we have
\begin{equation*}
    \begin{aligned}
        \{\cS_{\alpha}f_{\boldsymbol{\lambda}} g\}&=\sum_{\substack{i, j\in I \\ M, N \in \mathbb{Z}^D}} \frac{\partial g}{\partial u_N^j}(\boldsymbol{\lambda \cS})^N\{u_{\boldsymbol{\lambda \cS}}^i u^j\}(\boldsymbol{\lambda \cS})^{-M} \frac{\partial \cS_\alpha f}{\partial u_M^i}\\
        &=\sum_{\substack{i, j\in I\\ M, N \in \mathbb{Z}^D}} \frac{\partial g}{\partial u_N^j}(\boldsymbol{\lambda \cS})^N\{u_{\boldsymbol{\lambda \cS}}^i u^j\}(\boldsymbol{\lambda} \boldsymbol{\cS})^{-M+E_\alpha}\lambda_\alpha^{-1}\frac{\partial f}{\partial u_{M-E_\alpha}^i}
        =\lambda^{-1}_\alpha\{f_{\boldsymbol{\lambda}}g\},
    \end{aligned}
\end{equation*}
\begin{equation*}
    \begin{aligned}
        \{f_{\boldsymbol{\lambda}} \cS_{\alpha}g\}
        =S_{\alpha}\lambda_{\alpha}\sum_{\substack{i, j\in I \\ M, N \in \mathbb{Z}^D}} \frac{\partial g}{\partial u_{N-E_\alpha}^j}(\boldsymbol{\lambda \cS})^{N-E_\alpha}\{u_{\boldsymbol{\lambda S}}^i u^j\}(\boldsymbol{\lambda \cS})^{-M} \frac{\partial  f}{\partial u_M^i}
        =\lambda_{\alpha}\cS_{\alpha}\{f_{\boldsymbol{\lambda}}g\}.
    \end{aligned}
\end{equation*}

Moreover, for the left \textit{Leibniz rule} \eqref{1.2a} we have
\begin{equation*}
\left\{f_{\boldsymbol{\lambda}} g h\right\}=\sum_{j \in I, N \in \mathbb{Z}^D}\left(h \frac{\partial g}{\partial u^j_N}+g \frac{\partial h}{\partial u^j_N}\right)(\boldsymbol{\lambda \cS})^N\left\{f_{\boldsymbol{\lambda}} u^j\right\}=h\left\{f_{\boldsymbol{\lambda}} g\right\}+g\left\{f_{\boldsymbol{\lambda}} h\right\},
\end{equation*}
and similarly, for the \textit{right Leibniz rule} \eqref{1.2b}, we can use the second identity in (\ref{6}) to get:
\begin{equation*}
   \begin{aligned}
    \left\{f g_{\boldsymbol{\lambda}} h\right\}&=\sum_{i \in I, M \in \mathbb{Z}^D}{\left\{u^i_{\boldsymbol{\lambda \cS}} h\right\}}(\boldsymbol{\lambda \cS})^{-M}\left(\frac{\partial f}{\partial u^i_M} g+\frac{\partial g}{\partial u^i_M} f\right)\\
    &=\left\{f_{\boldsymbol{\lambda \cS}} h\right\} g+\left\{g_{\boldsymbol{\lambda \cS}} h\right\} f .
   \end{aligned}
\end{equation*}
To prove the equivalence between, respectively, the skewsymmetry and Jacobi identity for the generators of the PVA and for the full bracket we must rely on rather tedious computations. They can be found in Appendix~\ref{app:1} and \ref{app:2}.   
\end{proof}

\subsection{Multidimensional mPVAs and Hamiltonian structures}
Let $\triangle_\alpha=\cS_\alpha-1$. The elements of the quotient space   
\begin{align}\label{eq:locfunct-def}
    \mathcal{F} =\frac{\mathcal{A}}{{\bigtriangleup}_1 \mathcal{A}+{\bigtriangleup}_2 \mathcal{A}+\dots+{\bigtriangleup}_D \mathcal{A}} 
\end{align}
are called \emph{local functionals}. We denote the projection map from $\mathcal{A}$ to $\mathcal{F}$ as a formal integral, which associates to $f \in \mathcal{A}$ the elements $F :=\int f$ in $\mathcal{F}$. We will sometimes denote the equivalence relation as $ f\sim g$ if and only if $\int(f-g)=0$. Observe that, in particular, we have $\cS_\alpha f\sim f$.

The variational derivative of a local functional $F =\int f$ is defined as
\begin{equation}\label{eq:varder}
    \frac{\delta F}{\delta u^i}={\delta}_{u^i}F:=\sum_{N \in \mathbb{Z}^D} \boldsymbol{\cS}^{-N}\frac{\partial f}{\partial u^i_N} .
\end{equation}
\begin{proposition}\label{thm:var-welldef} The variational derivative \eqref{eq:varder} is well-defined in $\mc F$, namely $\delta _u(\int f)=\delta_u(\int g)$ if $\int(f-g)=0$.\end{proposition}
\begin{proof}
With a slight abuse of notation, we use the same symbol $\delta_u f$ to denote the RHS of \eqref{eq:varder} as an operator $\A\to\A$. First, we prove that for any $f \in \mathcal{A}$ and $\alpha=1,\ldots,D$, we have
    \begin{equation}\label{eq:kerTalpha}
        \frac{\delta}{\delta u^i}\bigtriangleup _{\alpha}f=0.
    \end{equation}
By definition of $\triangle_\alpha$ and of variational derivative,
\begin{align*}
        \frac{\delta}{\delta u^i}\bigtriangleup _{\alpha}f&= \frac{\delta}{\delta u^i}\cS_{\alpha}f- \frac{\delta}{\delta u^i}f\\
        &= \sum_{N \in \mathbb{Z}^D} \boldsymbol{\cS}^{-N}\frac{\partial \cS_{\alpha}f}{\partial u^i_N}-\sum_{N \in \mathbb{Z}^D} \boldsymbol{\cS}^{-N}\frac{\partial F}{\partial u^i_N}.
\end{align*}
Using \eqref{eq:commshift} in the first of the two summations, we can rewrite it as
$$
\sum_{N \in \mathbb{Z}^D} \boldsymbol{\cS}^{-N}\cS_{\alpha}\frac{\partial f}{\partial u^i_{N-E_{\alpha}}}=\sum_{N \in \mathbb{Z}^D} \boldsymbol{\cS}^{-(N-E_{\alpha})}\frac{\partial f}{\partial u^i_{N-E_{\alpha}}},
$$
which is manifestly equal to the second one; therefore, their difference vanishes. We have observed before that $\cS_\alpha f\sim f$, from which -- repeatedly applying the same property -- $\boldsymbol{\cS}^Mf\sim f$ for any $M=(m_1,m_2,\ldots,m_D)\in\Z^D$. To complete the proof of the proposition, we need to show that $\delta_{u^i}((\boldsymbol{\cS}^M-1)f)=0$. First of all, we observe that, for any $M$ and for any $\alpha$,
\begin{align*}
(\boldsymbol{\cS}^M-1)f&=(\cS_\alpha-1)\boldsymbol{\cS}^{M-E_\alpha}f+(\boldsymbol{\cS}^{M-E_\alpha}-1)f\\
&=-(\cS_\alpha-1)\boldsymbol{\cS}^Mf+(\boldsymbol{\cS}^{M+E_\alpha}-1)f.
\end{align*}
Applying the variational derivative to either side and taking \eqref{eq:kerTalpha} into account, we have $\delta_{u^i}(\boldsymbol{\cS}^M-1)f=\delta_{u^i}(\boldsymbol{\cS}^{M'}-1)f$, with $M'=M\pm E_\alpha$. This operation can be repeated a finite number of times for all the values of the index $\alpha$, until we get to $M^{(k)}=(0,\ldots,0)$, so $\delta_{u^i}(\boldsymbol{\cS}^M-1)f=\delta_{u^i}(1-1)f=0$.
\end{proof}

In the one-dimensional case, the Poisson bracket among local densities is related to a $\lambda$-bracket by the relation \eqref{eq:PoissBrack-1Dcase} \cite{10.1093/imrn/rny242}. Similarly, a multidimensional multiplicative PVA defines a Hamiltonian structure on the space $\mc F$ defined as \eqref{eq:locfunct-def}, given by 
\begin{equation}\label{eq:PoissBrack-def}
   \left \{\smallint f,\smallint g\right\}:=\int{\{f_{\boldsymbol{\lambda}}g\}}|_{\boldsymbol{\lambda}=1} \quad f,g\in \mathcal{A}.
\end{equation}
By $\mlambda=1$ we mean $(\lambda_1=1,\lambda_2=1,\ldots,\lambda_D=1)$. For simplicity, we introduce the notation 
\begin{equation}\label{4.2}
\{f,g\}={\{f_{\boldsymbol{\lambda}}g\}}|_{\boldsymbol{\lambda}=1}
\end{equation}
to denote the corresponding operation $\A\times\A\to\A$, disregarding the formal integration.
\begin{theorem}\label{thm:PVAHam}
    Let $\mathcal{A}$ be an algebra of differential polynomials with a $\lambda$-bracket 
    and consider the bracket on $\A$ defined in \eqref{4.2}. Then
    \begin{enumerate}[label=(\alph*)]
    \item The bracket \eqref{4.2}  induces a well-defined bracket on the quotient space $\mathcal{F}$;
    \item If the $\lambda$-bracket satisfies the axioms of a PVA, then the induced bracket
    on $\mathcal{F}$ is a Lie bracket.
    \item If the $\lambda$-bracket satisfies the axioms of a PVA, then the bracket \eqref{4.2} also induces a well-defined Lie algebra morphism from $\mathcal F$ to the evolutionary derivations of $\A$, given by
    \begin{align*}
    \smallint h\mapsto X_H\colon\A&\to \A\\
    f&\mapsto\{\smallint h,f\}
    \end{align*}
    such that $X_{\{F,G\}}=[X_F,X_G]$.
    \end{enumerate}
\end{theorem}
\begin{proof}
    Part (a). From the property of sesquilinearity we have that, for any $\alpha=1,...D$,
    \begin{align}\label{eq:sesqui1}
        \{f+(\cS_{\alpha}-1)h,g\}&=(\{f_{\boldsymbol{\lambda}}g\}+{\lambda}^{-1}_{\alpha}\{h_{\boldsymbol{\lambda}}g\}-\{h_{\boldsymbol{\lambda}}g\})|_{\boldsymbol{\lambda}=1}\\ 
        &=\{f,g\}
    \end{align}
    \begin{align*}
        \{f,g+(\cS_{\alpha}-1)h\}&=(\{f_{\boldsymbol{\lambda}}g\}+{\lambda}_{\alpha}\cS_{\alpha}\{f_{\boldsymbol{\lambda}}h\}-\{h_{\boldsymbol{\lambda}}h\})|_{\boldsymbol{\lambda}=1}\\
        &=\{f,g\}+(\cS_{\alpha}-1)\{f,h\}\sim \{f,g\}.
    \end{align*}
    Part (b). The Jacobi property for the bracket follows immediately by setting
    $\mlambda = \mmu = (1,\ldots,1)$ in the Jacobi identity for the PVA \eqref{Jacobi-identity}, while the skewsymmetry is a consequence of the
    skewsymmetry for the $\lambda$-bracket. Indeed, we have
    \begin{align*}
        \{g,f\}=\{g_{\boldsymbol{\lambda}}f\}|_{\boldsymbol{\lambda}=1}&=-{}_{\rightarrow}\{f_{(\boldsymbol{\lambda S})^{-1}}g\}|_{\boldsymbol{\lambda}=1}.\\
        \intertext{Recalling that in $\F$ all the shifts acting on the bracket are equivalent to no-shift,}
        &\sim-\{f_{(\boldsymbol{\lambda})^{-1}}g\}|_{\boldsymbol{\lambda}=1}\\
        &= -\{f,g\}.
    \end{align*}
 Part (c). First, observe that the bracket \eqref{4.2} is well defined in $\A$ when its first entry takes values in $\mathcal{F}$, as it follows from \eqref{eq:sesqui1}. Then, from \eqref{1.2a} we have that \eqref{4.2} is a derivation of its second argument, and from \eqref{1.1b} that it commutes with the shift operators, i.e. $\{\smallint h,\cS_\alpha f\}=\cS_\alpha\{\smallint h,f\}$. Finally, let us observe
 \begin{align*}
 X_{\{F,G\}}(h)&=\left\{\left\{\smallint f,\smallint g\right\},h\right\}&\left[X_F,X_G\right]&=\left\{\smallint f,\left\{\smallint g,h\right\}\right\}-\left\{\smallint g,\left\{\smallint f,h\right\}\right\}\\
 &=\left.\left\{\smallint\left\{f_{\mlambda}g\right\}\,{}_{\mlambda\mmu}h\right\}\right|_{\mlambda=1\atop\mmu=1}&&=\left.\left(\{f{}_{\mlambda}\{g_{\mmu}h\}\}-\{g{}_{\mmu}\{f{}_{\mlambda}h\}\}\right)\right|_{\mlambda = 1\atop\mmu = 1}
 \end{align*}
 On the one hand, the first expression is equal to the same one without the integral; and this latter one and the second expression are equal because of \eqref{Jacobi-identity}.
\end{proof}
Theorem~\ref{thm:PVAHam} establishes an equivalence between the notion of (multidimensional) multiplicative Poisson vertex algebra and local difference Hamiltonian operators that we will detail in the next Section. This is the main motivation for the introduction of the theory of multiplicative PVAs.
\section{Multidimensional scalar difference operators}
We denote by $\mathcal{M}_{\mathscr{l}}(\mathcal{A})=Mat_{\mathscr{l}\times \mathscr{l}}(\A[\boldsymbol{\cS},\boldsymbol{\cS}^{-1}])$ the algebra of (local) matrix difference
operators. Elements of $\mathcal{M}_{\mathscr{l}}(\mathcal{A})$ are Laurent polynomials
\begin{equation}\label{eq:op}
    P(\boldsymbol{\cS})=\sum_{l_1=m_1}^{n_1}\sum_{l_2=m_2}^{n_2}\cdots\sum_{l_D=m_D}^{n_D}P_{l_1,l_2,\ldots,l_D}\cS_1^{l_1}\cS_2^{l_2}\cdots\cS_D^{l_D}=\sum_{L\in [M,N]}P_L\boldsymbol{\cS}^L,
\end{equation}
$P_L=(P_L^{ij})_{i,j=1}^\ell\in Mat_{\ell\times\ell}(\A)$, with the associative product $\circ$ defined by the relation
\begin{equation}
    \cS\circ A=\cS (A)\cS, \quad A\in \mathcal{M}_{\mathscr{l}}(\mathcal{A}).
\end{equation}
We say that $P(\boldsymbol{\cS})$  as in \eqref{eq:op} is a difference operator of order $(A,B)$ if $\min m_\alpha=|M|=A$ and $\max n_\alpha=|N|=B$.  This reproduces and generalizes to the multi-dimensional case the standard definition for difference operators (see \cite{CW20}).

A difference operator of this class defines a bilinear map $\mathcal{F} \times \mathcal{F} \rightarrow \mathcal{F}$ by the formula 
\begin{equation}
\begin{split}
    B(F,G)&=\int \frac{\delta F}{\delta u} \cdot P\left(\frac{\delta G}{\delta u}\right)\\
    &=\int\sum_{i,j=1}^\ell\sum_{L\in\Z^D}\frac{\delta F}{\delta u^i} P^{ij}_L\boldsymbol{\cS}^L\left(\frac{\delta G}{\delta u^j}\right).
\end{split}
\end{equation}
The \textit{adjoint} of a difference operator $P(\cS)=(P^{ij}(\cS))$ of the form \eqref{eq:op} is the difference operator $P^*(\boldsymbol{\cS})=((P^*)^{ij}(\boldsymbol{\cS}))^\ell_{i,j=1}$, where
\begin{equation}
      (P^*)^{ij}(\cS)=\sum_{L\in[M,L]}\cS^{-L} \circ P^{ji}_L.
\end{equation}
If $P$ is skewadjoint, means $P^*(\cS)=-P(\cS)$, then the bilinear map $B$ is skewsymmetric and vice versa. A local difference operator of the form \eqref{eq:op} is skewadjoint if and only if $m_\alpha=-n_\alpha$ for alla $\alpha=1,\ldots,D$  and 
\begin{equation}
    P^{ij}_{-L}=-\cS^{-L}P^{ji}_L,
\end{equation}
for every $i,j=1,\dots, \mathscr{l}$, and $L$ between $(0,0,\ldots,0)$ and $N$. In particular, this means that $P_{\boldsymbol{0}}$ is a skewsymmetric matrix.

A  local Hamiltonian  ($\mathscr{l} \times \mathscr{l}$ ) matrix  difference operator $P^{i j}(\boldsymbol{\cS}) $ defines a multidimensional multiplicative PVA by letting
\begin{equation}
    (\{u^i_\lambda u^j\})^\mathscr{l}_{i,j=1}:=P^{ji}|_{\boldsymbol{\cS} \rightarrow \mlambda}=P^{ji}_L\mlambda^L,
\end{equation}
then extending the bracket from the generators of $\mathcal{A}$ to the full algebra according to properties \eqref{1.1a}-\eqref{1.2b} .
The expression for the bracket on the full algebra $\mathcal{A}$ is called the master formula (\cite{10.1093/imrn/rny242}) and it has the form:
\begin{equation}\label{eq:lambda-bracket-P}
    \{f_{\mlambda} g\}=\sum^\mathscr{l}_{i,j=1}\sum_{M,N\in \Z^D} \frac{\partial g}{\partial u^j_M}(\mlambda\boldsymbol{\cS})^M \{u^i_{\mlambda \boldsymbol{\cS}}u^j\}(\mlambda \cS)^{-N}\frac{\partial f }{\partial u_N^i} 
\end{equation}
where
\begin{equation}
    \{u^i_{\mlambda\boldsymbol{ \cS}} u^j\}=P^{ji}(\mlambda \boldsymbol{\cS})=\sum_{L\in[M,N]}P^{ji}_L\mlambda^L\boldsymbol{\cS}^L
\end{equation}
In particular, the condition of being an Hamiltonian operator for $P$ is equivalent the PVA-Jacobi identity for any triple of generators $(u^i,u^j,u^k)$.

Having established a convenient framework for the study and the manipulation of multidimensional difference Hamiltonian operators, now we focus on the study of \emph{scalar} operators, or equivalently multidimensional multiplicative Poisson vertex algebras generated by $\ell=1$ variable $u=u_{\boldsymbol{0}}$. In this case, we can denote the ``spatial indices'' with latin letters without ambiguity ($\boldsymbol{\cS}=\{\cS_i\}_{i=1}^D$).
Our first result is the classification of operators of order $(-1,1)$ and $(-2,2)$ in dimension 2.

\subsection{$(-1,1)$-order}
Generically, a $D$-dimensional, scalar skew-adjoint difference operator of order $(-1,1)$ is of the form
$$
P=\sum_{i=1}^D\left(f_i \cS_i-\cS_i^{-1}\circ f_i\right),
$$
where $f_i$, $i=1,\ldots,D$ are elements of $\A$. In this paragraph we focus on the case $D=2$; our result extends in a straightforward way to any $D>1$. More concretely, let us consider skew-adjoints operators of the general form
\begin{equation}\label{eq:op-1}
P=f(u,\ldots)\cS_1+g(u,\ldots)\cS_2-\cS_2^{-1}\circ g(u,\ldots)-\cS_1^{-1}\circ f(u,\ldots).
\end{equation}

\begin{proposition}\label{thm:ham1-2d}
A $(-1,1)$-order difference operator of the form \eqref{eq:op-1} is Hamiltonian if and only if it is of the form
\begin{equation}\label{eq:op-1-normal}
P=f(u)\left(\alpha \cS_1+\beta \cS_2-\beta \cS_2^{-1}-\alpha \cS_1^{-1}\right)\circ f(u)
\end{equation}
for arbitrary constant values of $\alpha,\beta$ and a generic function of one variable $f(u)$.
\end{proposition}
\begin{proof}
For compactness of notation, let $f_1=f$ and $f_2=g$; then the $\lambda$ bracket associated to \eqref{eq:op-1} is $\{u_{\lambda}u\}= \sum_{k=1}^2\left(\lambda_kf_k-\lambda_k^{-1}\cS^{-1}_kf_k\right)$. We write the Jacobi identity \eqref{Jacobi-identity} as the vanishing of the \emph{Jacobiator} $J(\lambda,\mu):=\{u_{\lambda}\{u_{\mu}u\}\}-\{u_{\mu}\{u_{\lambda}u\}\}-\{\{u_\lambda u\}_{\lambda \mu}u\}$. We compute $J(\lambda,\mu)$ using the master formula; for example, the first term of $J$ is
\begin{equation*}
\begin{split}
        \{u_{\lambda}\{u_{\mu}u\}\}&=\sum_{\substack{m,n\in \mathbb{Z}\\k,j=1,2}}\lambda_1^m\lambda_2^n\lambda_k(\cS_1^m\cS_2^nf_k)\frac{\partial}{\partial u_{mn}}\mu_jf_j\\
        &-\sum_{\substack{m,n\in \mathbb{Z}\\k,j=1,2}}\lambda_1^m\lambda_2^n\lambda_k^{-1}(\cS_1^m\cS_2^n\cS_k^{-1}f_k)\frac{\partial}{\partial u_{mn}}\mu_jf_j\\
        &-\sum_{\substack{m,n\in \mathbb{Z}\\k,j=1,2}}\lambda_1^m\lambda_2^n\lambda_j^{-1}\lambda_k\cS_j^{-1}(\cS_1^m\cS_2^nf_k)\frac{\partial}{\partial u_{mn}}\mu_j^{-1}f_j\\&+\sum_{\substack{m,n\in \mathbb{Z}\\k,j=1,2}}\lambda_1^m\lambda_2^n\lambda_j^{-1}\lambda_k^{-1}\cS_j^{-1}(\cS_1^m\cS_2^n\cS_k^{-1}f_k)\frac{\partial}{\partial u_{mn}}\mu_j^{-1}f_j.
\end{split}
\end{equation*}
The fulfillment of the Jacobi identity is then equivalent to the vanishing of all the coefficients of $(\lambda_1,\mu_1,\lambda_2,\mu_2)$ in $J$. Preliminarily, we have the following lemma. 
\begin{lemma}\label{lemma3-1}
    The Jacobi identity $J(\lambda,\mu)=0$ for a skewsymmetric bracket \eqref{eq:op-1} implies that $f=f(u,u_{10})$, $g=g(u,u_{01})$ .
\end{lemma}
\begin{proof}
    Assume $f_k\neq 0$ (if $f=f_1=0$ or $g=f_2=0$, we are in the same case as as one-dimensional scalar operators, studied in \cite{DSKVW18-1}). First of all, 
    let $m_k=max\{m|\frac{\partial f_k}{\partial u_{mn}}\neq 0\}$ and  suppose that $m_k\geq 2$.
     Computing the coefficient of $\lambda_1^{m_k+1}\lambda_2^n\mu_k$ in $(J=0)$, we obtain:$(\cS_1^{m_k}\cS_2^nf_k)\frac{\partial f_k}{\partial u_{m_kn}}-0=0$,
     hence $\frac{\partial f_k}{\partial u_{m_kn}}=0$, a contradiction. Hence $\frac{\partial f_k}{\partial u_{mn}}=0$ for $m>1$. In a similar way we prove that 
     $\frac{\partial(\cS_k^{-1} f_{k})}{\partial u_{mn}}=0$ for $m<-1$, namely $\frac{\dev f_k}{\dev u_{mn}}=0$ for $m<0$. In the same way we have $\frac{\partial f_k}{\partial u_{mn}}=0$ for $n>1$  and $n<0$. Therefore, a necessary condition for the vanishing of the Jacobiator is $f_k=f_k(u,u_{01},u_{10},u_{11})$. Computing the coefficient of $\lambda_1^{-2}\lambda_2^{-1}\mu_2$, we have 
     $(\cS_1^{-2}\cS_2^{-1}f_1)\frac{\partial g}{\partial u_{11}}=0$, so $g$ is not a functions in $u_{11}$; similarly for $f$. Finally, in a similar way we see that $f$ is not a functions in $u_{01}$ and
     $g$ is not a functions in $u_{10}$. This proves the lemma.
\end{proof}

With the results of Lemma \ref{lemma3-1}, the $\lambda$-bracket $\{u_\lambda u\}$ of (-1,1) order has the form
\begin{equation}\label{eq:op-3}
    \{u_\lambda u\}=(\lambda_1-(\lambda_1\cS_1)^{-1})\circ f(u,u_{10})+(\lambda_2-(\lambda_2\cS_2)^{-1})\circ g(u,u_{01}),
\end{equation}
Requiring the vanishing of the coefficients of $\lambda_1\mu_1^{-1}$ and $\lambda_2\mu_2^{-1}$ in the Jacobi identity $J=0$, we obtain the following equations on $f(u,u_{10})$ and $g(u,u_{01})$:
\begin{equation}\label{eq:op-4}
    -f\left(\cS_1^{-1}\frac{\partial f}{\partial u_{10}}\right)+(\cS_1^{-1} f)\frac{\partial f}{\partial u}=0,
\end{equation}
\begin{equation}\label{eq:op-5}
    -g\left(\cS_2^{-1}\frac{\partial g}{\partial u_{01}}\right)+(\cS_2^{-1} g)\frac{\partial g}{\partial u}=0.
\end{equation}
First, by \eqref{eq:op-4},
\begin{equation}\label{eq:op-6}
    \frac{\partial_{u} f}{f}=\cS_1^{-1}\left( \frac{\partial_{u_{10}} f}{f}\right)
\end{equation}
Since the LHS (resp. RHS) of this equation is a function of $u, u_{10}$ (resp. of $u_{-10}, u$), we conclude that both sides are functions of
$u$ only. We have 
\begin{align*}
    \frac{\partial_{u} f}{f}&=k(u),\\
    \frac{\partial_{u_{10}} f}{f}&=k(u_{10}).
\end{align*}
Hence
\begin{equation*}
    f=\varphi _1(u)\varphi _2(u_{10})
\end{equation*}
It follows that $k(u)=\varphi _1^{'}(u)/\varphi_1(u)$ and $\cS_1^{-1}k(u_{10})=\cS_1^{-1}(\varphi _2^{'}(u_{10})/\varphi_2(u_{10}))=\varphi _2^{'}(u)/\varphi_2(u)$.
Thus
\begin{align*}
    \frac{\varphi _1^{'}(u)}{\varphi_1(u)}&=\frac{\varphi _2^{'}(u)}{\varphi_2(u)}\\
\Rightarrow \frac{d}{du}\log(\varphi _1)&=\frac{d}{du}\log(\varphi _2)\\
\Rightarrow \varphi _2&=c\varphi _1
\end{align*}
where $c$ is a constant.
We have:
\begin{equation}\label{eq:op1-ansf}
    f=\varphi(u)\varphi(u_{10}).
\end{equation}
Note that $\varphi(u_{10})=\cS_1\varphi(u)$, so we can use the shorthand notation $f=\varphi\varphi_{10}$. By \eqref{eq:op-5}, in the same way, we have 
\begin{equation}\label{eq:op1-ansg}
    g=\psi\psi_{01}.
\end{equation}
Replacing the ansatz  \eqref{eq:op1-ansf} and \eqref{eq:op1-ansg} in the Jacobi identity we find that the 24 nonvanishing terms (which appear as coefficients of $\lambda_1^a\lambda_2^b\mu_1^c\mu_2^d$ with $(a,b,c,d)$ taking some values $(-1,0,1)$) are all of the form $\varphi_{mn}\psi_{pq}\left(\varphi\psi'-\psi\varphi'\right)_{rs}$. A necessary and sufficient condition on $\varphi$ and $\psi$ is therefore
$\psi=c\, \varphi$ for some arbitrary constant $c$. The form \eqref{eq:op-1-normal} is obtained by rescaling the function $\varphi(u)=:\alpha f(u)$, $c\,\alpha=:\beta$ to obtain a more symmetric result; however, note that the operator depends on one arbitrary function and \emph{one} constant.
\end{proof}
\begin{remark}
Note that, for $f$ a positive (respectively negative) function, there always exists a change of variables taking $P$ to the constant form $P'=\alpha \cS_1+\beta \cS_2-\beta \cS_2^{-1}-\alpha \cS_1^{-1}$. Indeed, let $v(u)=\int^u\frac{1}{f(s)}\ud s$. Then, by the master formula we have
\begin{equation}\label{eq:1stOrder-const}
\begin{split}
P'=&\{v_{\mlambda} v\}|_{\mlambda\to\boldsymbol{\cS}}=v'\,\{u_{\mlambda\boldsymbol{\cS}}u\}\,v'|_{\mlambda\to\boldsymbol{\cS}}\\
&=\frac{1}{f(u)}f(u)\left.\left(\alpha\cS_1 \left(f(u)\frac{1}{f(u)}\right)\lambda_1+\beta\cS_2 \left(f(u)\frac{1}{f(u)}\right)\lambda_2-\cdots\right)\right|_{\mlambda\to\boldsymbol{\cS}}\\
&=\left(\alpha\cS_1+\beta\cS_2-\beta\cS_2^{-1}-\alpha\cS_1^{-1}\right).
\end{split}
\end{equation}
\end{remark}
\begin{remark}
One might be tempted to consider a constant two-dimensional operator of the form \eqref{eq:1stOrder-const} as an essentially one-dimensional operator on a deformed lattice, similarly to the differential case \cite{CCS17}. However, observe that a linear combination of shift operators is \emph{not} a shift operator, so the analogy does not work. Indeed, a linear combination of shift operators is not an algebra endomorphism: for $\widetilde{S}:=(\alpha S_1+\beta S_2)$ we have $\widetilde{S}(fg)=\alpha f_1g_1+\beta f_2g_2$, while $(\widetilde{\cS}f)(\widetilde{\cS}g)=\alpha^2f_1g_1+\beta^2f_2g_2+\alpha\beta(f_1g_2+f_2g_1)$.
\end{remark}

\subsection{$(-2,2)$-order}
We consider skew-adjoint operators of $(-2,2)$ order, of the generic form
\begin{equation}\label{eq:generic}
    P(\cS)=\sum_{i,j=1}^D\left(F_{ij}(u,\ldots)\cS_i \cS_j-\cS_i^{-1}S_j^{-1}\circ F_{ij}(u,\ldots)\right)+\sum_{i=1}^D\left(G_i(u,\ldots)\cS_i-\cS_i^{-1}\circ G_i(u,\ldots)\right)
\end{equation}
\begin{proposition}\label{prop:2D}
For $D=2$, all the scalar Hamiltonian operators of order $(-2,2)$ are either of the form 
\begin{equation}\label{eq:2ndOrderNormal}
P=f(u)\left(\alpha \cS_1^2+\beta \cS_1\cS_2+\gamma \cS_2^2+\delta S_1 +\epsilon \cS_2-(\cS_i\leftrightarrow \cS_i^{-1})\right)\circ f(u)
\end{equation}
for arbitrary values of the constants (with $\alpha,\beta,\gamma$ not all vanishing) or
\begin{align}\label{eq:2ndOrder1D}
P_i&=f\cS_i\circ F^\alpha\cS_i\circ f+f\left(\left(F^\alpha+\beta\right)\cS_i+\cS_i\circ\left(F^\alpha+\beta\right)\right)\circ f\\ \notag
&-f\left(\left(F^\alpha+\beta\right)\cS^{-1}_i+\cS^{-1}_i\circ\left(F^\alpha+\beta\right)\right)\circ f-f\cS^{-1}_i\circ F^\alpha\cS^{-1}_i\circ f
\end{align}
for $i=1,2$ (no summation), $\alpha$ and $\beta$ arbitrary constants ($\alpha\neq0$, or else this is a subcase of \eqref{eq:2ndOrderNormal}), $f=f(u)$, and $F(u)$ a non-constant solution of $fF'=F$.
\end{proposition}
As for the first-order case, the first step of the proof is identifying the variable dependency of the coefficients $F_{ij}$ and $G_i$. As before, we compute the $J(\lambda,\mu)$ for the $\lambda$ bracket associated to \eqref{eq:generic}, i.e. $\{u_\lambda u\}:=P|_{\cS_{1,2}\to\lambda_{1,2}}$, and collect the coefficients of the variables $\lambda,\mu$. Our preliminary result is the following

\begin{lemma}\label{lemma3-2}
    The Jacobi identity for a skew-adjoint operator of the form \eqref{eq:generic}  imply that $F_{11}=F_{11}(u,u_{10},u_{20})$, $F_{12}=F_{12}(u,u_{10},u_{01},u_{11})$, $F_{22}=F_{22}(u,u_{01},u_{02})$, $G_1=G_1(u,u_{10})$,  $G_2=G_2(u,u_{01})$.
\end{lemma}
\begin{proof}
    We assume that $F_{ij}\neq0$ at least for some $(i,j)=1,2$ (if all $F_{ij}=0$, the bracket is not of order $(-2,2)$). Let $m_{ij}=max\{m|\frac{\partial F_{ij}}{\partial u_{mn}}\neq0\}$, and suppose that $m_{ij}\geq3$. Calculating the coefficient of $\lambda^{m_{ij}+1}_1\lambda_2^n\mu_i\mu_j$ in $J(\lambda,\mu)$ we obtain $\frac{\partial F_{ij}}{\partial u_{m_{ij} n}}=0$. Hence $\frac{\partial F_{ij}}{\partial u_{mn}}=0$ for $m>2$. Similarly, we prove $\frac{\partial\left(\cS_i^{-1}\cS_j^{-1} F_{ij}\right)}{\partial u_{mn}}=0$ for $m<-2$, which is equivalent to $\frac{\partial F_{ij}}{\partial u_{mn}}=0$ for $m<0$.  In the same way, we have $\frac{\partial F_{ij}}{\partial u_{mn}}=0$ for $n>2$ and $n<0$. We also have $\frac{\partial G_{i}}{\partial u_{mn}}=0$ for $m>1$ or $n>1$. We can first conclude that necessary conditions for the Jacobi identity are $F_{ij}=F_{ij}(u,u_{01},u_{10},u_{11},u_{02},u_{20},u_{21},u_{12},u_{22})$ and $G_i=G_i(u,u_{01},u_{10},u_{11})$.  To identify a stricter dependency on the variables, we perform direct computations with help of a computer algebra system. We use it to identify coefficients of the Jacobi identity whose vanishing implies the non-dependency of $F_{ij}$ or $G_i$ by some of the variables. For example, the coefficient of $\lambda_1^{4}\lambda_2^2\mu_1^2$ in $J(\lambda,\mu)$ is $(\cS_1^2\cS_2^2F_{11})\frac{\dev F_{11}}{\dev u_{22}}$, implying $\frac{\dev F_{11}}{\dev u_{22}}=0$. Continuing in this way, we conclude that $F_{11}$ is not a function of $u_{01},u_{11} ,u_{02},u_{12},u_{21}$; $F_{12}$ is not a function of $u_{20} ,u_{02},u_{12},u_{21}$; $F_{22}$ is not a function of $u_{10},u_{11} ,u_{20},u_{12},u_{21}$;$ G_{1}$ is not a function of $u_{01},u_{11}$, and $G_{2}$ is not a function of $u_{10},u_{11}$. This proves the Lemma.
\end{proof}
\begin{proof}[Proposition \ref{prop:2D}]

Let $F_{ij}$ and $G_i$ have the variable dependency prescribed by Lemma \ref{lemma3-2}; to proceed further we must consider the different cases for the vanishing of some of the leading terms $F_{ij}$
\begin{description}
    \item[Generic case: $F_{ij}\neq0$ for all $(i,j)$]~\\
    From the coefficient of $\lambda_1^2\mu_1\mu_2^2$ we have $\cS_1F_{22}\frac{\partial F_{11}}{\partial u_{10}}=0$, so $F_{11}$ is not a function of $u_{10}$. In same way, we get $F_{12}=F_{12}(u,u_{11})$, $F_{22}=F_{22}(u,u_{02})$, $G_1=G_1(u,u_{10})$, $G_2=G_2(u,u_{01})$.
With these assumptions, many of the coefficients in $J(\lambda,\mu)$ take an easy form. From the coefficients of $\lambda_1^4\mu_1^2$, $\lambda_1^2\lambda_2^2\mu_1\mu_2$, $\lambda_2^4\mu_2^2$, $\lambda_1^3\lambda_2\mu_1^2$ , $\lambda_1^3\lambda_2\mu_1\mu_2$, $\lambda_1^2\lambda_2^2\mu^2_1$, $\lambda_1^2\lambda_2^2\mu_2^2$, respectively, we have:
\begin{equation}\label{eq:op-9}
    (\cS_1^2 F_{11})\frac{\partial F_{11}}{\partial u_{20}}-F_{11}\cS_1^2\left(\frac{\partial F_{11}}{\partial u_{}}\right)=0,
\end{equation}
\begin{equation}\label{eq:op-10}
    (\cS_1\cS_2 F_{12})\frac{\partial F_{12}}{\partial u_{11}}-F_{12}\cS_1\cS_2\left(\frac{\partial F_{12}}{\partial u_{}}\right)=0,
\end{equation}
\begin{equation}\label{eq:op-11}
    (\cS_2^2 F_{22})\frac{\partial F_{22}}{\partial u_{02}}-F_{22}\cS_2^2\left(\frac{\partial F_{22}}{\partial u_{}}\right)=0,
\end{equation}
\begin{equation}\label{eq:op-12}
    (\cS_1^2 F_{12})\frac{\partial F_{11}}{\partial u_{20}}-F_{11}\cS_1^2\left(\frac{\partial F_{12}}{\partial u_{}}\right)=0,
\end{equation}
\begin{equation}\label{eq:op-13}
    (\cS_1\cS_2 F_{11})\frac{\partial F_{12}}{\partial u_{11}}-F_{12}\cS_1\cS_2\left(\frac{\partial F_{11}}{\partial u_{}}\right)=0,
\end{equation}
\begin{equation}\label{eq:op-14}
    (\cS_1^2 F_{22})\frac{\partial F_{11}}{\partial u_{20}}-F_{11}\cS_1^2\left(\frac{\partial F_{22}}{\partial u_{}}\right)=0,
\end{equation}
\begin{equation}\label{eq:op-15}
    (\cS_2^2 F_{11})\frac{\partial F_{22}}{\partial u_{02}}-F_{22}\cS_2^2\left(\frac{\partial F_{11}}{\partial u_{}}\right)=0.
\end{equation}
By Equation \eqref{eq:op-9}, we get
\begin{equation}\label{eq:op-16}
    \frac{\partial_{u_{20}} F_{11}}{F_{11}}=\cS_1^{2}\left( \frac{\partial_{u_{}} F_{11}}{F_{11}}\right).
\end{equation}
Since the LHS (resp. RHS) of this equation is a function of $u, u_{20}$ (resp. of $u_{20}, u_{40}$), we conclude that both sides are functions of
$u_{20}$ only. We have 
\begin{align*}
    \frac{\partial_{u} F_{11}}{F_{11}}&=k_{1}(u),\\
    \frac{\partial_{u_{20}} F_{11}}{F_{11}}&=k_{1}(u_{20}).
\end{align*}
Hence
\begin{equation*}
    F_{11}=f_1(u)f_2(u_{20})
\end{equation*}
Moreover, $k_{1}(u)=f_1^{'}(u)/f_1(u)=f'_2(u)/f_2(u)$. Analogously to the previous section, this latter relation can be integrated to obtain $f_2=\alpha f_1$, where $\alpha$ is a constant. Then, denoting $f_1(u)=f(u)$, we have
\begin{equation}\label{eq:op-17}
    F_{11}=\alpha\, f(u)f(u_{20}).
\end{equation}
By \eqref{eq:op-10} and \eqref{eq:op-11}, in a similar way, we obtain
\begin{equation}\label{eq:op-18}
    F_{12}=h(u)h(u_{11}),\qquad F_{22}=g(u)g(u_{02}).
\end{equation}
Next, we replace \eqref{eq:op-17} and \eqref{eq:op-18} in Equation \eqref{eq:op-12} and \eqref{eq:op-13}:
\begin{equation}
    \frac{f'(u_{20})}{f(u_{20})}=\cS_1^2\left(\frac{h'(u)}{h(u)}\right), \quad \frac{h'(u_{11})}{h_{11}}=\cS_1\cS_2\left(\frac{f'(u)}{f(u)}\right)
\end{equation}
This gives us $F_{12}=\beta\, f(u)f(u_{11})$; doing the same with \eqref{eq:op-14} and \eqref{eq:op-15} we also conclude $F_{22}=\gamma\, f(u)f(u_{02})$.

Again substituting the expressions for $F_{11}$, $F_{12}$, and $F_{22}$ in the Jacobi identity, the vanishing of the coefficients of $\lambda_1^3\mu^2_1$, $\lambda_1^3\mu^1,\lambda_1^2\lambda_2\mu_1^2$, and $\lambda_1^2\lambda_2\mu_2$ is the following set of equations:
\begin{align}\label{eq:op-19}
    f(u)(\cS_1^2 G_1)f'(u_{20})-f(u)f(u_{20})\cS_1^2\left(\frac{\partial G_1}{\partial u}\right)&=0,\\
    G_1\cS_1(f(u_{20})f'(u))-\frac{\partial G_1}{\partial u_{10}}\cS_1(f(u)f(u_{20}))&=0,\\
    f(u)f'(u_{20})\cS_1^2 G_2-f(u)f(u_{20})\cS_1^2\left(\frac{\partial G_2}{\partial u}\right)&=0,\\
    -G_2\cS_2(f(u_{20})f'(u))+\frac{\partial G_2}{\partial u_{01}}\cS_2 (f(u)f(u_{20}))&=0.
\end{align}
We proceed similarly as we did to solve \eqref{eq:op-12}-\eqref{eq:op-15} and we obtain $G_1=\delta\, f(u)f(u_{10})$. $G_2=\epsilon\, f(u)f(u_{01})$, with $\delta$ and $\epsilon$ arbitrary constants; it is now immediate to write \eqref{eq:2ndOrderNormal} by observing that, e.g., $f(u_{10})=\cS_1 f(u)$.
    \item[Case 2: $F_{22}=0$]
The $\lambda$ bracket associated to the difference operator is, in this case, 
\begin{equation}\label{eq:op-20}
   \begin{split}
        \{u_\lambda u\}=&(\lambda_1^2-(\lambda_1\cS_1)^{-2}) F_{11}+(\lambda_1\lambda_2-(\lambda_1\lambda_2\cS_1\cS_2)^{-1}) F_{12}\\
        &+(\lambda_1-(\lambda_1\cS_1)^{-1}) G_1+(\lambda_2-(\lambda_2\cS_2)^{-1}) G_2.
   \end{split}
\end{equation}
We proceed similarly to the previous case, finding necessary conditions from the vanishing of different coefficients of the monomials in $(\lambda_1,\lambda_2,\mu_1,\mu_2)$. For instance, we do not have the coefficient $\lambda_1^2\mu_1\mu_2^2$, so we cannot assume from the beginning that $F_{11}=F_{11}(u,u_{20})$. However, the coefficient of $\lambda_1^4\mu_1^2$ is of the same form as in the generic case, requiring $F_{11}=f(u)g(u_{10})f(u_{20})$ and $F_{12}=f(u)h_1(u_{10})h_2(u_{01})f(u_{11})$. Then the coefficients of $\lambda_1^3\mu_1^2\mu_2^{-1}$ and $\lambda_1^3\mu_1$ require that $h_1(u_{10})$ and $h_2(u_{01})$ are constants, and so on until we get $F_{11}=\alpha \,f(u)f(u_{20})$, $F_{12}=\beta f(u)f(u_{11})$, $G_1=\delta f(u)f(u_{10})$, $G_2=\epsilon f(u)f(u_{01})$; therefore, this is simply a particular case of \eqref{eq:2ndOrderNormal}. The same type of result holds true if $F_{11}$ is the only vanishing $F_{ij}$ term.
\item[Case 3: $F_{12}=F_{22}=0$]
First, note that this case is equivalent, upon exchanging of $\cS_1$ and $\cS_2$, to the one when $F_{22}$ is the only \emph{nonvanishing} $F_{ij}$ term.

When looking for terms allowing us to identify the functional dependence of $F_{11}$, $G_1$, and $G_2$ we encounter two subcases: after observing that necessary conditions for the Jacobi identities are $F_{11}=F_{11}(u,u_{10},u_{20})$, $G_1=G_1(u,u_{10})$, and $G_2=G_2(u,u_{01})$, we notice that the coefficient of $\lambda_1^2\mu_1\mu_2$ is $-(\cS_1G_2)\dev_{u_{10}}F_{12}$; the solutions of the system is radically different depending on the vanishing of $G_2$.

If $G_2\neq0$, we are back to Case 2 with $\beta=0$, namely $F_{12}$ is vanishing. However, if $G_2=0$ then all the shifts are only along one direction, which means that the associated $\lambda$ bracket is of the form
\begin{equation}\label{eq:op-21}
    \{u_\lambda u\}=F_{11}\lambda_1^2+G_1\lambda_1-(\cS_1\lambda_1)^{-2}F_{11}-(\cS_1\lambda_1)^{-1}G_1,
\end{equation}
with $F_{11}=F_{11}(u,u_{10},u_{20})$ and $G_1=G_1(u,u_{10})$.

Computing the Jacobi identity of \eqref{eq:op-21}, we first obtain $F_{11}=f(u)F(u_{10})f(u_{20})$, from which the  vanishing of the coefficient of $\lambda_1^3\mu_1^2$ requires 
\begin{equation}\label{eq:op-22}
  \begin{split}
        &f(u)F(u_{10})f'(u_{20})\cS_1^2 G_1+f(u)f(u_{10})f(u_{20})f(u_{30})F(u_{20})F'(u_{10})\\&=f(u)f(u_{20})F(u_{10})\cS_1^2\left(\frac{\partial G_1}{\partial u}\right)\\
        \Rightarrow &f'(u_{20})\cS_1^2 G_1+\frac{f(u_{10})f(u_{20})f(u_{30})F(u_{20})F'(u_{10})}{F(u_{10})}=f(u_{20})\cS_1^2\partial_uG_1
  \end{split}
\end{equation}
Since the first term in the LHS and the RHS are dependent only on $u_{20}$ and $u_{30}$, we conclude that the second term in the LHS is independent of $u_{10}$. Hence,
\begin{equation}\label{eq:op-23}
  \begin{split}
       &\frac{f(u_{10})F'(u_{10})}{F(u_{10})}=a\quad\mathrm{const.}\\
       \Rightarrow&fF'=aF.
  \end{split}
\end{equation}
First, consider the case where $a\neq0$. Then we can see from \eqref{eq:op-22} that
\begin{equation}
f'(u_{20})\cS_1^2G_1+af(u_{20})f(u_{30})F(u_{20})=f(u_{20})\cS_1^2\partial_uG_1,
\end{equation}
or, equivalently,
\begin{equation}
a\,f(u)\,F(u)\,f(u_{10})=f(u)\partial_uG_1(u,u_{10})-f'(u)G_1(u,u_{10}).
\end{equation}
Dividing both sides by $f^2(u)$, we obtain:
\begin{equation}
   \begin{split}
        \frac{af(u_{10})F(u)}{f(u)}&=\frac{f(u)\partial_uG_1(u,u_{10})-G_1(u,u_{10})f'(u)}{f^2(u)}\\
        \stackrel{\text{\eqref{eq:op-23}}}{\Rightarrow} f(u_{10})F'(u)&=\frac{\partial}{\partial u}\left(\frac{G_1(u,u_{10})}{f(u)}\right).
   \end{split}
\end{equation}
Integrating by $u$ and multiplying by $F(u)$, we obtain:
\begin{equation}\label{eq:op-24}
    G_1(u,u_{10})=f(u)f(u_{10})F(u)+f(u)A(u_{10}),
\end{equation}
where $A$ is an arbitrary function of $u_{10}$. Next, we consider the coefficient of $\lambda_1^3\mu_1$ in the Jacobi identity, whose vanishing condition is, after the substitution of \eqref{eq:op-23} and \eqref{eq:op-24},
\begin{equation}
    \begin{split} &f(u_{10})F(u_{20})A'(u_{10})=A(u_{10})F(u_{20})f'(u_{10})+f(u_{10})f(u_{20})F(u_{10})F'(u_{20}),\\
   \Rightarrow  &f(u_{10})A'(u_{10})-A(u_{10})f'(u_{10})=\frac{f(u_{10})f(u_{20})F(u_{10})F'(u_{20})}{F(u_{20})}  \\
   \Rightarrow&f(u_{10})A'(u_{10})-A(u_{10})f'(u_{10})=af(u_{10})F(u_{10})\\
   \Rightarrow&\frac{\partial}{\partial u}\left(\frac{A(u)}{f(u)}\right)=a\left(\frac{F(u)}{f(u)}\right).
    \end{split}
\end{equation}
Integrating by $u$, we obtain
\begin{equation}
    A(u_{10})=F(u_{10})f(u_{10})+cf(u_{10}),
\end{equation}
i.~e.~
\begin{equation}\label{eq:op-25}
    G_1(u,u_{10})=f(u)f(u_{10})(F(u_{})+F(u_{10})+c).
\end{equation}
Recall that $F$ satisfies equation \eqref{eq:op-23} for a generic value of $a$; this means that, if we choose a particular value for the constant and solve for function $F$, any power of $F$ can be used to define a Hamiltonian operator:
\begin{equation}
 f\,\left(F^\alpha\right)'=\alpha\,F^{\alpha-1}fF'=\alpha F^{\alpha-1}\, a F\Rightarrow f(F^\alpha)'=a\alpha F^\alpha.
\end{equation}
This means that we can characterize $F$ as a solution of $fF'=F$ and write the lambda bracket \eqref{eq:op-21} as 
\begin{equation}
\begin{split}
    \{u_{\mlambda} u\}&=f\,(\cS_1 F^\alpha)\,(\cS_1^2f)\lambda_1^2+f(\cS_1f)\left(F^\alpha+(\cS_1 F^\alpha)+c\right)\lambda_1\\&-(\cS_1^{-1}f)\,f\left(F^\alpha+(\cS_1^{-1}F^\alpha)+ c\right)\lambda_1^{-1}-(\cS_1^{-2}f)(\cS_1^{-1}F^\alpha)\,f\lambda_1^{-2},
\end{split}
\end{equation}
with $\alpha$ and $c$ arbitrary constants, $f=f(u)$, and $F(u)$ a non-constant solution of $fF'=F$. A simple redefinition of $c=\beta$ allows us to write the operator associated to the bracket as in \eqref{eq:2ndOrder1D}.
\item[Case 4: $F_{11}=F_{22}=0$]
In this case, by computing the Jacobi identity, we first get $F_{12}=f(u)h(u_{10},u_{01})f(u_{11}),G_1=g(u)g(u_{10}),G_2=k(u)k(u_{01})$. The computations differ if $G_i$'s are nonvanishing or not, but the conclusion is always the same: $h(u_{10},u_{01})$ must be a constant and $f$, $g$ and $k$ must be proportional. Therefore, this is yet another operator of the form \eqref{eq:2ndOrderNormal} with $\alpha=\gamma=0$, and possibly vanishing of $\delta$ and $\epsilon$.
\end{description}
\end{proof}
\begin{remark}
The brackets \eqref{eq:2ndOrder1D} are essentially one-dimensional; the proof that they are all the possible solutions for (1-dimensional) second-order Hamiltonian difference operators can be found in \cite{DSKVW18-1} -- notice that the power $\alpha$ of the functions $F$ can either be arbitrarily taken, or set equal to 1 by requiring that $fF'=\frac1\alpha F$. On the other hand, in the language of \cite{DSKVW18-1}, the appearance of the constant $\beta$ is due to the fact that ``any multiplicative Poisson $\lambda$-bracket of order $\leq 2$ is either of general type or a linear combination of the $\lambda$-bracket of complementary type and the $\lambda$ bracket of general type of order 1'' (Theorem 2.5), so our result matches with De Sole et al.'s one.
\end{remark}
\section{Bi-Hamiltonian pairs}
It is well-known that, when there exist two different (but ``compatible'') Hamiltonian formulations for the same system of equations we say that the system is \emph{bi-Hamiltonian}. The bi-Hamiltonian formulation of a system is a key criterion to determine its integrability and obtain its symmetries \cite{magri, olver}.

We say that two Hamiltonian operators $P$ and $Q$ are compatible, and therefore form a bi-Hamiltonian pair, if and only if any linear combinations of the two $P_\eta=P+\eta Q$ is a Hamiltonian operator for any $\eta$.

In the one-dimensional case, De Sole, Kac, Valeri and Wakimoto identified several compatible Hamiltonian pairs \cite{DSKVW18-1}. Considering only $(-1,1)$ and $(-2,2)$-order operators, we have two possible cases:
\begin{enumerate}
\item $P$ and $Q$ are both of the same so-called \emph{general form}, namely $P$ and $Q$ are both of the form \begin{equation}
P_i=f(u)\left(\alpha_i \cS^2+\beta_i\cS-\beta_i\cS^{-1}-\alpha_i\cS^{-2}\right)\circ f(u)
\end{equation}
for a function $f(u)$ and arbitrary constants $(\alpha_i,\beta_i)$, $i=1,2$, $P=:P_1$, $Q=:P_2$.
\item $P$ is of the form \eqref{eq:2ndOrder1D} and $Q$ is of order $(-1,1)$ and of the form $Q=\kappa f(u)(\cS-\cS^{-1})\circ f(u)$ for the same function $f(u)$.
\end{enumerate}
It is a natural question to ask whether, in the multidimensional case, more rigid constraints appear or the picture is the same. To this aim, we can start from the result of Proposition \ref{prop:2D}, since we can observe that all the $(-1,1)$-order Hamiltonian operator are of the form \eqref{eq:2ndOrderNormal} for $\alpha=\beta=\gamma=0$ and therefore we do not need to consider them as a separate case. We explore the following three cases:
\begin{enumerate}
    \item $P$ and $Q$ both of normal form \eqref{eq:2ndOrderNormal};
    \item $P$ of normal form and $Q$ of one-dimensional form \eqref{eq:2ndOrder1D};
    \item $P$ and $Q$ both of one-dimensional form \eqref{eq:2ndOrder1D}, either along the same direction or in orthogonal directions.
\end{enumerate}
The following results are obtained by direct computation of the Jacobi identity for $P_\eta=P+\eta Q$, where $P$ and $Q$ are as above and $\eta$ is a formal parameter. Such lengthy but straightforward computations can be easily performed with a computer algebra system.
\begin{proposition}\label{prop:biHamNN}
   Two two-dimensional, scalar Hamiltonian operators of the form \eqref{eq:2ndOrderNormal} are compatible if and only if $f_2(u)=\kappa f_1(u)$, for a nonvanishing constant $\kappa$, where $f_1$ denotes the function $f$ in the definition of $P$ and $f_2$ the one in the definition of $Q$.
\end{proposition}
\begin{proof} It is immediate to check that the ansatz $f_2(u)=\kappa f_1(u)$ (here $X_i$ denotes the parameter $X$ (either function or constant)  used in \eqref{eq:2ndOrderNormal} for the definition of the operator $P_{1,2}=P,\,Q$) is a sufficient condition to have a bi-Hamiltonian pair. To prove that such a condition is also necessary, first we notice that, for $f_1$ not proportional to $f_2$, the coefficient of $\lambda_1^4\mu_1^2$ in the Jacobiator requires that $\alpha_1\alpha_2=0$; if $\alpha_1$ is nonvanishing, further terms require that $\beta_2=\gamma_2=\delta_2=\epsilon_2=0$, therefore $\alpha_1$ must be vanishing too. However, the coefficient of $\lambda_1^2\lambda_2^2\mu_1\mu_2$ can vanish only if $\beta_1\beta_2=0$\dots again, if $\beta_1$ is nonvanishing all $\beta_2$, $\gamma_2$, $\delta_2$ and $\epsilon_2$ must be zero. This pattern continues until we rule out the possibility that $f_2$ is not proportional to $f_1$, because that would imply that $\alpha_2=\beta_2=\gamma_2=\delta_2=\epsilon_2=0$, i.e. $P_2=0$.
\end{proof}
\begin{proposition}\label{prop:biHamN1}
    Two two-dimensional scalar Hamiltonian $(P,Q)$ operators of the form, respectively, \eqref{eq:2ndOrderNormal} and \eqref{eq:2ndOrder1D} are compatible if and only if $f_2=\kappa f_1$, $\alpha=\beta=\gamma=0$, and either $\epsilon=0$ or $\delta=0$ depending on whether $Q$ is of the form $P_1$ or $P_2$ of \eqref{eq:2ndOrder1D}. 
\end{proposition}
In short, this reproduces the one-dimensional picture of De Sole et al. \cite{DSKVW18-1}
\begin{proof}
Without loss of generality, let us consider $Q$ of the form $P_1$ of \eqref{eq:2ndOrder1D}. Calculating the Jacobiator for $P_\eta$, we observe that the coefficient of $\lambda_1^4\mu_1^2$ vanishes either for $f_1f_2'-f_2f_1'=0$ or for $\alpha=0$. In the former case, whose solution is $f_2(u)=\kappa f_1(u)$, we find that the vanishing of $\alpha$, $\beta$, $\gamma$ and $\epsilon$ are, respectively, necessary conditions for the vanishing of the coefficients of $\lambda_1^3\mu_1^2$, $\lambda_1^2\lambda_2\mu_1^2$, $\lambda_1^2\mu_1\mu_2^2$, and $\lambda_1^2\lambda_2\mu_1\mu_2^{-1}$, and sufficient conditions for the vanishing of full Jacobiator.

Alternatively, if we do not impose the proportionality between $f_1$ and $f_2$ we have that necessary conditions for the vanishing of $\lambda_1^4\mu_1^2$, $\lambda_1^3\lambda_2\mu_1^2$, $\lambda_1^2\lambda_2^2\mu_1^2$, and  $\lambda_1^2\mu_1^2\mu_2^{-1}$ are, respectively, $\alpha=\beta=\gamma=\epsilon=0$. Imposing them, the vanishing of the coefficient of $\lambda_1^3\mu_1^2$ for $f_1$ not proportional to $f_2$ requires $\delta=0$, hence $P=0$.
\end{proof}

\begin{proposition}\label{prop:biHam11}
Two two-dimensional scalar Hamiltonian operators of the form \eqref{eq:2ndOrder1D} with $\alpha\neq 0$ are compatible if and only if the shift operators are in the same direction and $f_2=\kappa f_1$, $\alpha_2=\kappa \alpha_1$ for an arbitrary constant $\kappa$.
\end{proposition}
\begin{proof}
First of all, we rule out the possibility that $P=P_1$ and $Q=P_2$, namely that the two essentially one-dimensional operators are in different directions. This can be seen as a by-product of Proposition \ref{prop:2D}: the operators $P_i$ are individually Hamiltonian, but their sum $P_1+P_2$ is not. $P_1$ and $P_2$, therefore, are never compatible. This is confirmed by explicit computations of the Jacobi identity: we first find that a necessary condition is $f_2=\kappa f_1$ and then that $\alpha_1=\alpha_2=0$.

On the other hand, let us consider $P=P_1$ with $f=f_1$, $\alpha=\alpha_1$, $\beta=\beta_1$ and $Q=P_1$ with different parameters $f=f_2$, $\alpha=\alpha_2$, $\beta=\beta_2$.  The vanishing of the coefficient $\lambda_1^4\mu_1^2$ in the Jacobiator requires as necessary condition $f_2=\kappa f_1$; moreover, after this ansatz we have that $F_2=\rho\;F_1^{1/\kappa}$ for arbitrary constants $\kappa$ and $\rho$ and that the Jacobiator vanishes if and only if $\alpha_2=\kappa \alpha_1$. Note that the Poisson pencil $P_\eta$ in this case would correspond to the operator $P_i$ with the rescaling $f(u)\mapsto \sqrt{1+\eta\,\kappa^2}f(u)$ and $\alpha\mapsto\alpha/\sqrt{1+\eta\,\kappa^2}$.
\end{proof}

Recall from \cite{DSKVW18-1} that, given a bi-Hamiltonian pair of difference operators, the Lenard-Magri scheme is an algorithm to construct a hierarchy of conserved quantities in involution. Such a hierarchy is integrable if the scheme produces infinitely many independent conserved quantities. To construct the hierarchy, one normally starts with a Casimir function of one of the two Hamiltonian structures, but it is crucial that such a Casimir is not a Casimir of the second structure too. From the results of Proposition \ref{prop:biHamNN}, \ref{prop:biHamN1}, and \ref{prop:biHam11}, we readily see that among all the possible two-dimensional, scalar bi-Hamiltonian pairs for which the Casimir functions do not coincide are the pairs of the form
$P=f(u)(\cS_i-\cS_i^{-1})\circ f(u)$ and $Q=P_i$ as in \eqref{eq:2ndOrder1D}. The Casimir function of $P$ is $\int h_{-1}=\int^u f^{-1}(s)\ud s$ which is not, as it can be immediately veerified, a Casimir of $Q$. These bi-Hamiltonian pairs are one-dimensional and have been already studied in \cite{DSKVW18-1}; their corresponding integrable hierarchy has been proved to be (always equivalent to) the  Volterra hierarchy.

\section{Conclusive remarks}
\quad In this work, we addressed the study of multidimensional multiplicative Poisson Vertex Algebras, extended to the $D>1$ case of the multiplicative $\lambda$-bracket, later we give the Master formula and proved the Jacobi identity of the multiplicative $\lambda$-bracket, which is similar to the $D=1$ case. We define the local functionals, and show the relation of the PVAs and Hamiltonian structures.
Then we classify the 2-dimensional scalar difference operators with the order $(-2,2)$. A Hamiltonian difference operator of $(-1,1)$-order
 is of the form
 \begin{equation*}
 P=f(u)\left(\alpha \cS_1+\beta \cS_2-\beta \cS_2^{-1}-\alpha \cS_1^{-1}\right)\circ f(u)
 \end{equation*}
It is straightforward to generalize the result of Proposition \ref{thm:ham1-2d} to arbitrary $D$, for which we have
$$
P=\sum_{i=1}^D \alpha_i f(u)\left(S_i-S_i^{-1}\right)\circ f(u).
$$
All the the two-dimensional scalar Hamiltonian operators of order $(-2,2)$ are either of the form 
\begin{equation*}
P=f(u)\left(\alpha \cS_1^2+\beta \cS_1\cS_2+\gamma \cS_2^2+\delta S_1 +\epsilon \cS_2-(\cS_i\leftrightarrow \cS_i^{-1})\right)\circ f(u)
\end{equation*}
for arbitrary values of the constants (with $\alpha,\beta,\gamma$ not all vanishing) or
\begin{align*}
P_i&=f\cS_i\circ F^\alpha\cS_i\circ f+f\left(\left(F^\alpha+\beta\right)\cS_i+\cS_i\circ\left(F^\alpha+\beta\right)\right)\circ f\\ \notag
&-f\left(\left(F^\alpha+\beta\right)\cS^{-1}_i+\cS^{-1}_i\circ\left(F^\alpha+\beta\right)\right)\circ f-f\cS^{-1}_i\circ F^\alpha\cS^{-1}_i\circ f
\end{align*}
Moreover, for the two-dimensional case, two scalar Hamiltonian operators $P $and $Q$ are compatible if and only if the functions $f_1$ and $f_2$ (which definition the operators $P$ and $Q$) are proportional.

The generalization of this result to $D>2$ case is straightforward, too -- we anticipate that a generic Hamiltonian operator of order $(-2,2)$ is of the form 
\begin{equation}
    P=\sum_{i,j=1}^Df(u)(\alpha_{ij}\cS_i\cS_j+\beta_i\cS_i-\beta_i\cS_i^{-1}-\alpha_{ij}\cS_i^{-1}\cS_j^{-1})\circ f(u)
\end{equation}
for arbitrary constants $\alpha_{ij}=\alpha_{ji}$, $\beta_i$ and a function $f(u)$.

Operators of this ``general'' form (the name has been firsr used in \cite{10.1093/imrn/rny242}) exist for arbitrary $(-N,N)$ order operators, as sums of homogeneous operators
\begin{equation*}
   P_k=f(u)\left(\sum_{i_1,\ldots,i_k=1}^D\alpha^{(k)}_{i_1i_2\cdots i_k}\left(\cS_{i_1}\cS_{i_2}\ldots\cS_{i_k}-\left(\cS_{i_1}\cS_{i_2}\ldots\cS_{i_k}\right)^{-1}\right)\right)\circ f(u).
\end{equation*}
However, further investigations are required to rule out the existence of operators of "complementary" (or otherwise exceptional) type.

The non-local version of multiplicative Poisson vertex algebra has been defined in \cite{DSKVW19-9}; the key difference is the existence of a non-local $\lambda$-bracket, $\{-_ \lambda -\}: \mathcal{A}\otimes \mathcal{A \rightarrow}\mathcal{A}[[\lambda,\lambda^{-1}]]$, taking value in formal series of $\lambda$ and $\lambda^{-1}$. The main idea is quite simple: it is possible to represent a rational difference operator, i.e. the ratio of two local difference operators, by infinite series expansion of the denominators (mapping $\A(\lambda)\hookrightarrow \A[[\lambda,\lambda^{-1}]]$). The authors use the language of non-local multiplicative PVAs to generalize the construction of the q-deformed $W$ algebra introduced by Frenkel and Reshetikhin in \cite{FR96}; they obtain the bi-Hamiltonian structure and conjecture that the corresponding hierarchy is integrable. The generalization of their construction to the multidimensional case is an open problem with important applications: for instance, the series of papers by Wiersma and Capel \cite{WIERSMA0, WIERSMA1, WIERSMA2, WIERSMA3} presents several examples of nonlocal two-dimensional Hamiltonian structures for integrable D$\Delta$Es, albeit of a rather simple form.

 In the same way as the theory of PVA proved itself a versatile tool and an effective language for the study of Hamiltonian integrable PDEs, in this work we have showed that multiplicative PVAs enjoy similar advantages in the category of multidimensional differential-difference equations.

 The two notions of PVA and mPVA are connected, as detailed in the recent work by one of the authors \cite{CV25}, where an explicit mapping between difference and differential algebra is proposed. Outside the context of PVAs, a similar argument is well-known and presented, for example, in \cite{CDZ}. Moreover, following the ideas of Kuperschmidt, it is possible to introduce more general PVAs of ``mixed'' type, that can be used to describe systems of functions depending on a set of discrete and continuous independent variables. For example, consider the case of the (semi) continuous limit of the 2D Toda equation presented in \cite{WIERSMA3}; the integrable hierarchy for a function $\psi(\{t_i\},x,n)$ -- $\{t_i\}$ are the hierarchy times, $x$ a continuous variable and $n$ a discrete variable of the original two-dimensional lattice -- is Hamiltonian with respect to the operator 
\begin{equation}
  \partial_x^{-1}\left(\mathcal{S}+\mathcal{S}^{-1}-2\right),
\end{equation}
where $\cS$ denotes the shift in the discrete direction. The natural generalization of both standard and multiplicative PVAs that would correspond to an operator of this type is obtained by a $\lambda$-bracket with two formal variables $\lambda|\mu$, such that (for example) $\{\psi_{\lambda|\mu}\cS\psi\}=\mu\cS\{\psi_{\lambda|\mu}\psi\}$ and $\{\psi_{\lambda|\mu}\dev\psi\}=(\lambda+\dev)\{\psi_{\lambda|\mu}\psi\}$. The investigation and formalization of this idea will be discussed in a future work.

\paragraph*{Acknowledgments}
M.~C.~is a member of the GNFM INdAM group and he is supported by the National Science Foundation of China (Grants Nos. 12101341 and 12431008) and by Ningbo University High Level Science and Technology Projects Cultivation Plan. Y.~P.~ is supported by Ningbo University Second-Class Academic Scholarship (Upper-division) and National Graduate Student Stipend. Both the authors are supported by Ningbo City Yongjiang Innovative Talent Program and Ningbo University Talent Introduction and Research Initiation Fund.

\appendix \section{Proof of Theorem~\ref{theorem1}: skewsymmetry and Jacobi identity}
\subsection{Skewsymmetry}\label{app:1}
Consider the expression for $\{f_{\mlambda}g\}$ according to the master formula \eqref{master}:
    \begin{equation}\label{eq:provaskew1}
        \begin{aligned}
        \{f_{\boldsymbol{\lambda}} g\}&=\sum_{\substack{i, j\in I \\ M, N ,T\in \mathbb{Z}^D}} \frac{\partial g}{\partial u_N^j}(\boldsymbol{\lambda \cS})^N\circ P^{ji}_{T}(\mlambda\boldsymbol{\cS})^{T}(\boldsymbol{\lambda \cS})^{-M} \frac{\partial f}{\partial u_M^i}\\
        &=\sum_{\substack{i, j\in I \\ M, N ,T\in \mathbb{Z}^D}} \frac{\partial g}{\partial u_N^j}(\boldsymbol{\cS}^N P^{ji}_{T})\boldsymbol{\cS}^{N+T-M} \frac{\partial f}{\partial u_M^i}\boldsymbol{\lambda}^{N+T-M}.
        \end{aligned}
    \end{equation}
Here, we denote with $(\boldsymbol{\cS}^M f)$ the operator $\boldsymbol{\cS}^M$ acting only on the difference function $f$, while by $(\mlambda\boldsymbol{\cS})^M\circ f\boldsymbol{\cS}^N$ the composition of difference operators. From \eqref{eq:provaskew1}, the expression on the RHS of the skewsymmetry property \eqref{3} can be expanded as
    \begin{align*}       
            {}_{\rightarrow}\{f_{(\boldsymbol{\lambda \cS})^{-1}}g\}&=
            \sum_{\substack{i, j\in I \\ M, N ,T\in \mathbb{Z}^D}}{(\boldsymbol{\lambda \cS})}^{-(N+T-M)}\circ \frac{\partial g}{\partial u_N^j}(\boldsymbol{\cS}^N P^{ji}_{T})\boldsymbol{\cS}^{N+T-M} \frac{\partial f}{\partial u_M^i}\\
            &=\sum_{\substack{i, j\in I \\ M, N ,T\in \mathbb{Z}^D}}  \frac{\partial f}{\partial u_M^i}(\boldsymbol{\lambda \cS})^M\circ (\boldsymbol{\cS}^{-T} P^{ji}_{T})(\boldsymbol{\lambda \cS})^{-T}(\boldsymbol{\lambda \cS})^{-N}\frac{\partial g}{\partial u_N^j}.\\
            \intertext{Observe now that the skewsymmetry of the bracket between the generators is the same as the identity \eqref{eq:skewgen}, so that}
             {}_{\rightarrow}\{f_{(\boldsymbol{\lambda \cS})^{-1}}g\}&=-\sum_{\substack{i, j\in I \\ M, N ,T\in \mathbb{Z}^D}}  \frac{\partial f}{\partial u_M^i}(\boldsymbol{\lambda \cS})^M \circ P^{ij}_{-T}(\boldsymbol{\lambda \cS})^{-T}\cdot(\boldsymbol{\lambda \cS})^{-N}\frac{\partial g}{\partial u_N^j}\\
            &=-\sum_{\substack{i, j\in I \\ M, N \in \mathbb{Z}^D}}  \frac{\partial f}{\partial u_M^i}(\boldsymbol{\lambda   \cS})^M \{u^j_{\boldsymbol{\lambda \cS}}u^i\}(\boldsymbol{\lambda \cS})^{-N}\frac{\partial g}{\partial u_N^j}
            =-\{g_{\boldsymbol{\lambda}} f\}.
    \end{align*}
    
\subsection{Jacobi identity}\label{app:2}
First, we prove that the Jacobi identity holds true for any triple of \emph{shifted} generators $(u^i_M,u^j_N,u^k_P)$. Using the sesquilinearity property, we have
\begin{align}\notag
\{{u^i_M}_{\boldsymbol{\lambda}}\{{u^j_N}_{\boldsymbol{\mu}}u^k_P\}\} &=\{{u^i_M}_{\boldsymbol{\lambda}}{(\boldsymbol{\mu \cS })^P}\{{u^j}_{\boldsymbol{\mu}}u^k\}{\boldsymbol{\mu}}^{-N}\}\\\notag
&={\boldsymbol{\mu}^{P-N}}{(\boldsymbol{\lambda \cS})^P} \{u^i_{\boldsymbol{\lambda}}\{u^j_{\boldsymbol{\mu}}u^k\}\}{\boldsymbol{\lambda}^{-M}}\\\label{eq:proveJacobi1}
&={\boldsymbol{\mu}^{P-N}}{\boldsymbol{\lambda}^{P-M}}{\boldsymbol{\cS}^P} \{u^i_{\boldsymbol{\lambda}}\{u^j_{\boldsymbol{\mu}}u^k\}\},\\\label{eq:proveJacobi2}
\{{u^j_N}_{\boldsymbol{\mu}}\{{u^i_M}_{\boldsymbol{\lambda}}u^k_P\}\}& ={\boldsymbol{\lambda}^{P-M}}{\boldsymbol{\mu}^{P-N}}{\boldsymbol{\cS}^P} \{u^j_{\boldsymbol{\mu}}\{u^i_{\boldsymbol{\lambda}}u^k\}\},\\\notag
\{{\{{u^i_M}_{\boldsymbol{\lambda}}{u^j_N}\}}_{\boldsymbol{\lambda \mu}}u^k_P\} &=\{[(\boldsymbol{\lambda \cS})^N\{{u^i}_{\boldsymbol{\lambda}}{u^j}\}{\boldsymbol{\lambda}}^{-M}]_{\boldsymbol{\lambda \mu}}u^k_P\}\\\notag
        &={\boldsymbol{\lambda}^{N-M}}{(\boldsymbol{\lambda \mu \cS})^P} \{{\{u^i_{\boldsymbol{\lambda}}u^j\}}_{\boldsymbol{\lambda \mu}}u^k\}{\boldsymbol{\lambda}^{-N}}{\boldsymbol{\mu}^{-N}}\\\label{eq:proveJacobi3}
        &={\boldsymbol{\lambda}^{P-M}}{\boldsymbol{\mu}^{P-N}}{\boldsymbol{\cS}^P}\{{\{u^i_{\boldsymbol{\lambda}}u^j\}}_{\boldsymbol{\lambda \mu}}u^k\}.
\end{align}
We see that all the expressions \eqref{eq:proveJacobi1}, \eqref{eq:proveJacobi2} and \eqref{eq:proveJacobi3} share a common prefactor $\mlambda^{P-M}\mmu^{P-N}\boldsymbol{\cS}^P$, so that their combination $\text{\eqref{eq:proveJacobi1}}- \text{\eqref{eq:proveJacobi2}}-\text{\eqref{eq:proveJacobi3}}$ vanishes if and only if the Jacobi identity between the generators holds true. We then move on, to prove that the Jacobi identity holds true in general.

According to (\ref{2.2}) and (\ref{6}) we have
\begin{equation}\label{3.1}
    \begin{aligned}
        \left\{f_{\boldsymbol{\lambda}} g\right\}&=\sum_{\substack{i, j\in I \\ M, N \in \mathbb{Z}^D}} \frac{\partial g}{\partial u^j_N}{\left\{{u^i_M}_{\boldsymbol{\lambda \cS}} u^j_N\right\}}\frac{\partial f}{\partial u^i_M}\\ &=
    \sum_{\substack{i, j\in I \\ M, N,T \in \mathbb{Z}^D}} \frac{\partial g}{\partial u^j_N} ({\boldsymbol{\cS}}^{N+T-M}\frac{\partial f}{\partial u^i_M}){{\left\{{u^i_M}_{\boldsymbol{\lambda }} u^j_N\right\}}}_{(N+T-M)}
    \end{aligned}
\end{equation}
where we denote $\{{u^i_M}_{\boldsymbol{\lambda}} u^j_N\}={\{{u^i_M}_{\boldsymbol{\lambda }} u^j_N\}}_{(N+T-M)}\mlambda^{N+T-M}$.

Then, direct computations show
\begin{equation}
    \begin{aligned}
        \{f_{\boldsymbol{\lambda}}\{g_{\boldsymbol{\mu}}h\}\}&=\sum_{j \in I, N \in \mathbb{Z}^D}\{f_{\boldsymbol{\lambda}}{\{{u^j_N}_{\boldsymbol{\mu \cS}}h\}}\frac{\partial g}{\partial u^j_N} \}\\&=
        \sum_{\substack{j\in I \\  N,T,P \in \mathbb{Z}^D}}{\{{u^j_N}_{\boldsymbol{\mu}}h\}}_{(P+T-N)}\{f_{\boldsymbol{\lambda}}({\boldsymbol{\cS}}^{P+T-N}\frac{\partial g}{\partial u^j_N})\}
       \\ &+\sum_{\substack{j,k\in I \\  N,T,P \in \mathbb{Z}^D}}({\boldsymbol{\cS}}^{P+T-N}\frac{\partial g}{\partial u^j_N})\{f_{\boldsymbol{\lambda}}\frac{\partial h}{\partial u^k_P}{\{{u^j_N}_{\boldsymbol{\mu}}u^k_P\}}_{(P+T-N)}\}\label{3.2}
    \end{aligned}
\end{equation}
The first term in the RHS of (\ref{3.2}) can be rewritten, using the identity in (\ref{3.1}), as
\begin{equation}
    \begin{aligned}
        &\sum_{\substack{j\in I \\  N,T,P \in \mathbb{Z}^D}}{\{{u^j_N}_{\boldsymbol{\mu}}h\}}_{(P+T-N)}\{f_{\boldsymbol{\lambda}}({\boldsymbol{\cS}}^{P+T-N}\frac{\partial g}{\partial u^j_N})\}\\
        &=\sum_{\substack{j\in I \\  N,T,P \in \mathbb{Z}^D}}{\{{u^j_N}_{\boldsymbol{\mu}}h\}}_{(P+T-N)}({\boldsymbol{\lambda \cS}})^{P+T-N}\{f_{\boldsymbol{\lambda}}\frac{\partial g}{\partial u^j_N}\}
        \\ &=\sum_{j\in I,N \in \mathbb{Z}^D}{\{{u^j_N}_{\boldsymbol{\lambda \mu \cS}}h\}}\{f_{\boldsymbol{\lambda}}\frac{\partial g}{\partial u^j_N}\}
    \end{aligned}
\end{equation}
The second term in the RHS of (\ref{3.2}) can be rewritten as 
\begin{align}
        &\sum_{\substack{j,k\in I \\  N,T,P \in \mathbb{Z}^D}}({\boldsymbol{\cS}}^{T_1}\frac{\partial g}{\partial u^j_N})\{f_{\boldsymbol{\lambda}}\frac{\partial h}{\partial u^k_P}{\{{u^j_N}_{\boldsymbol{\mu}}u^k_P\}}_{(T_1)}\}\notag\\
        &=\sum_{\substack{j,k\in I \\  N,T,P \in \mathbb{Z}^D}}({\boldsymbol{\cS}}^{T_1}\frac{\partial g}{\partial u^j_N})\frac{\partial h}{\partial u^k_P}\{f_{\boldsymbol{\lambda}}{\{{u^j_N}_{\boldsymbol{\mu}}u^k_P\}}_{(T_1)}\}
        +\sum_{\substack{j,k\in I \\  N,T,P \in \mathbb{Z}^D}}({\boldsymbol{\cS}}^{T_1}\frac{\partial g}{\partial u^j_N}){\{{u^j_N}_{\boldsymbol{\mu}}u^k_P\}}_{(T_1)}\{f_{\boldsymbol{\lambda}}\frac{\partial h}{\partial u^k_P}\}\notag\\
        &=\sum_{\substack{i,j,k\in I \\  N,T,P \in \mathbb{Z}^D}}({\boldsymbol{\cS}}^{T_2}\frac{\partial f}{\partial u^i_M})({\boldsymbol{\cS}}^{T_1}\frac{\partial g}{\partial u^j_N})\frac{\partial h}{\partial u^k_P}{\{{u^i_M}_{\boldsymbol{\lambda}}{\{{u^j_N}_{\boldsymbol{\mu}}u^k_P\}}_{(T_1)}\}}_{(T_2)}
        +\sum_{\substack{k,l\in I \\  P,Q \in \mathbb{Z}^D}}(\frac{\partial^2 h}{\partial u^k_P \partial u^l_Q}){\{{g}_{\boldsymbol{\mu}}u^k_P\}}\{f_{\boldsymbol{\lambda}} u^l_Q\}\notag
\end{align}
where $T_1=P+T-N, T_2=P+T-M$.
The second term in the LHS of (\ref{Jacobi-identity}) is the same as the first term, after exchanging $f$ with $g$ and $\lambda$ with $\mu$. Therefore, we get that the LHS of (\ref{Jacobi-identity}) is
\begin{align*}
    &\left\{f_{\boldsymbol{\lambda}}\left\{g_{\boldsymbol{\mu}} h\right\}\right\}-\left\{g_{\boldsymbol{\mu}}\left\{f_{\boldsymbol{\lambda}} h\right\}\right\}\\
    &=\sum_{\substack{i,j,k\in I \\  N,T,P \in \mathbb{Z}^D}}({\boldsymbol{\cS}}^{T_2}\frac{\partial f}{\partial u^i_M})({\boldsymbol{\cS}}^{T_1}\frac{\partial g}{\partial u^j_N})\frac{\partial h}{\partial u^k_P}({\{{u^i_M}_{\boldsymbol{\lambda}}{\{{u^j_N}_{\boldsymbol{\mu}}u^k_P\}}\}}-{\{{u^j_N}_{\boldsymbol{\mu}}{\{{u^i_M}_{\boldsymbol{\lambda}}u^k_P\}}\}})\notag\\  
    &+\sum_{j\in I,N \in \mathbb{Z}^D}{\{{u^j_N}_{\boldsymbol{\lambda \mu \cS}}h\}}\{f_{\boldsymbol{\lambda}}\frac{\partial g}{\partial u^j_N}\}
    -\sum_{i\in I,M \in \mathbb{Z}^D}{\{{u^i_M}_{\boldsymbol{\lambda \mu \cS}}h\}}\{g_{\boldsymbol{\mu}}\frac{\partial f}{\partial u^i_M}\}\\
\end{align*}
we next look at the RHS of (\ref{Jacobi-identity}).
\begin{align}
    &\{\left\{f_{\boldsymbol{\lambda}} g\right\}_{\boldsymbol{\lambda \mu}} h\}=\sum_{\substack{i,j\in I \\  M,N \in \mathbb{Z}^D}}\{{(\frac{\partial g}{\partial u^j_N}{\{{u^i_M}_{\boldsymbol{\lambda \cS}}{u^j_N}\}}\frac{\partial f}{\partial u^i_M})}_{\boldsymbol{\lambda \mu}}h\}\notag\\
    &=\sum_{\substack{i,j\in I \\  M,N \in \mathbb{Z}^D}}{\{{\frac{\partial g}{\partial u^j_N}}_{\boldsymbol{\lambda \mu \cS}}h\}}{\{f_{\boldsymbol{\lambda}}{u^j_N}\}}
    +\sum_{\substack{i,j,k\in I \\  N,T,P \in \mathbb{Z}^D}}({\boldsymbol{\cS}}^{T_1}\frac{\partial g}{\partial u^j_N}){\{{({\boldsymbol{\cS}}^{T_3}\frac{\partial f}{\partial u^i_M})\{{u^i_M}_{\boldsymbol{\lambda}}{u^j_N}\}}_{\boldsymbol{\lambda \mu}}h\}}\notag\\
    &=\sum_{j\in I,N \in \mathbb{Z}^D}{\{{u^j_N}_{\boldsymbol{\lambda \mu \cS}}h\}}_{}\{f_{\boldsymbol{\lambda}}\frac{\partial g}{\partial u^j_N}\}
    +\sum_{\substack{i,j,k\in I \\  N,T,P \in \mathbb{Z}^D}}({\boldsymbol{\cS}}^{T_2}\frac{\partial f}{\partial u^i_M})({\boldsymbol{\cS}}^{T_1}\frac{\partial g}{\partial u^j_N})\frac{\partial h}{\partial u^k_P}{\{{\{{u^i_M}_{\boldsymbol{\lambda}}{u^j_N}\}}_{\boldsymbol{\lambda \mu}}u^k_P\}}\notag\\
    &+\sum_{\substack{i,j,k\in I \\  N,T,P \in \mathbb{Z}^D}}{\{{({\boldsymbol{\cS}}^{T_3}\frac{\partial f}{\partial u^i_M})}_{\boldsymbol{\lambda \mu \cS}}h\}}(\frac{\partial g}{\partial u^j_N}\{{u^i_M}_{\boldsymbol{\lambda}}{u^j_N}\})\notag\\\label{3.4}
\end{align}
where $T_3=N+T-M$ , and the last term in the RHS of (\ref{3.4}) is
\begin{align*}
    &\sum_{\substack{i,j\in I \\  M,N,T \in \mathbb{Z}^D}}{\{{({\boldsymbol{\cS}}^{T_3}\frac{\partial f}{\partial u^i_M})}_{\boldsymbol{\lambda \mu \cS}}h\}}(\frac{\partial g}{\partial u^j_N}\{{u^i_M}_{\boldsymbol{\lambda}}{u^j_N}\})\\
    &=\sum_{\substack{i,l\in I \\  M,N,T,Q \in \mathbb{Z}^D}}({\boldsymbol{\lambda \mu \cS}})^{-T_3}{\{{u^i_M}_{\boldsymbol{\lambda \mu \cS}}h\}}_{}(\frac{\partial^2 f}{\partial u^i_M \partial u^l_Q}{\{{u^l_Q}_{\boldsymbol{\lambda}}g\}}_{(T_3)})
    \\&=\sum_{\substack{i\in I \\  M,N,T \in \mathbb{Z}^D}}{\{{u^i_M}_{\boldsymbol{\lambda \mu \cS}}h\}}_{}({\boldsymbol{\lambda \mu \cS}})^{-T_3}{\{{\frac{\partial f}{\partial u^i_M}}_{\boldsymbol{\lambda}}g\}}_{(T_3)}\\
    &=\sum_{i\in I,M \in \mathbb{Z}^D}{\{{u^i_M}_{\boldsymbol{\lambda \mu \cS}}h\}}_{}{\{{\frac{\partial f}{\partial u^i_M}}_{(\boldsymbol{\mu \cS})^{-1}}g\}}
    =-\sum_{i\in I,M \in \mathbb{Z}^D}{\{{u^i_M}_{\boldsymbol{\lambda \mu \cS}}h\}}_{}\{g_{\boldsymbol{\mu}}\frac{\partial f}{\partial u^i_M}\}.
\end{align*}
We can finally compare the LHS and the RHS of the Jacobi identity: after the cancellation of the identical terms, we are left with
\begin{multline}
\{f_{\mlambda}\{g{}_{\mmu}h\}\}-\{g_{\mmu}\{f{}_{\mlambda}h\}\}-\{\{f_{\mlambda}g\}{}_{\mlambda\mmu}h\}\}=\\
\sum_{\substack{i,j,k\in I \\  N,T,P \in \mathbb{Z}^D}}({\boldsymbol{\cS}}^{T_2}\frac{\partial f}{\partial u^i_M})({\boldsymbol{\cS}}^{T_1}\frac{\partial g}{\partial u^j_N})\frac{\partial h}{\partial u^k_P}\left(\{u^i_M{}_{\mlambda}\{u^j_N{}_{\mmu}u^k_P\}\}-\{u^j_N{}_{\mmu}\{u^i_M{}_{\mmu}u^k_P\}\}-{\{{\{{u^i_M}_{\boldsymbol{\lambda}}{u^j_N}\}}_{\boldsymbol{\lambda \mu}}u^k_P\}}\right),
\end{multline}
where all the terms in the parenthesis vanish because, as proved above, the Jacobi identity holds true for any triple of shifted generators.
\bibliographystyle{plain}
\bibliography{reference}

\end{document}